\documentclass[11pt]{article}

\usepackage{amsthm}
\usepackage{graphicx} 
\usepackage{array} 

\usepackage{amsmath, amssymb, amsfonts, verbatim}
\usepackage{hyphenat,epsfig,subcaption,multirow}
\usepackage[font=small,labelfont=bf]{caption}
\usepackage{nicefrac}
\usepackage{paralist}

\usepackage[usenames,dvipsnames]{xcolor}
\usepackage[ruled]{algorithm2e}

\DeclareFontFamily{U}{mathx}{\hyphenchar\font45}
\DeclareFontShape{U}{mathx}{m}{n}{
      <5> <6> <7> <8> <9> <10>
      <10.95> <12> <14.4> <17.28> <20.74> <24.88>
      mathx10
      }{}
\DeclareSymbolFont{mathx}{U}{mathx}{m}{n}
\DeclareMathSymbol{\bigtimes}{1}{mathx}{"91}

\usepackage{tcolorbox}
\tcbuselibrary{skins,breakable}
\tcbset{enhanced jigsaw}

\usepackage[normalem]{ulem}
\usepackage[compact]{titlesec}

\definecolor{DarkRed}{rgb}{0.5,0.1,0.1}
\definecolor{DarkBlue}{rgb}{0.1,0.1,0.5}

\usepackage{nameref}
\definecolor{ForestGreen}{rgb}{0.1333,0.5451,0.1333}
\definecolor{Red}{rgb}{0.9,0,0}
\usepackage[linktocpage=true,
	pagebackref=true,colorlinks,
	linkcolor=DarkRed,citecolor=ForestGreen,
	bookmarks,bookmarksopen,bookmarksnumbered]
	{hyperref}
\usepackage[noabbrev,nameinlink]{cleveref}
\crefname{property}{property}{Property}
\creflabelformat{property}{(#1)#2#3}
\crefname{equation}{eq}{Eq}
\creflabelformat{equation}{(#1)#2#3}

\usepackage{bm}
\usepackage{url}
\usepackage{xspace}
\usepackage[mathscr]{euscript}

\usepackage{tikz}
\usetikzlibrary{arrows}
\usetikzlibrary{arrows.meta}
\usetikzlibrary{shapes}
\usetikzlibrary{backgrounds}
\usetikzlibrary{positioning}
\usetikzlibrary{decorations.markings}
\usetikzlibrary{patterns}
\usetikzlibrary{calc}
\usetikzlibrary{fit}
\usetikzlibrary{decorations}

\usepackage[framemethod=TikZ]{mdframed}

\usepackage[noend]{algpseudocode}
\makeatletter
\def\BState{\State\hskip-\ALG@thistlm}
\makeatother

\usepackage{cite}
\usepackage{enumitem}
\setlist[itemize]{leftmargin=20pt}
\setlist[enumerate]{leftmargin=20pt}

\usepackage[margin=1in]{geometry}

\usepackage{thmtools}
\usepackage{thm-restate}

\newtheorem{theorem}{Theorem}
\newtheorem{lemma}{Lemma}[section]
\newtheorem{proposition}[lemma]{Proposition}

\newtheorem{claim}[lemma]{Claim}
\newtheorem{fact}[lemma]{Fact}

\newtheorem{invariant}[lemma]{Invariant}

\newtheorem*{claim*}{Claim}
\newtheorem*{assumption*}{Assumption}
\newtheorem*{proposition*}{Proposition}
\newtheorem*{lemma*}{Lemma}

\newtheorem{observation}[lemma]{Observation}

\newtheorem*{theorem*}{Theorem}

\crefname{lemma}{Lemma}{Lemmas}
\crefname{claim}{claim}{claims}
\crefname{property}{Property}{Properties}
\crefname{invariant}{Invariant}{Invariants}

\newtheorem{mdresult}{Result}
\newenvironment{result}{\begin{mdframed}[backgroundcolor=lightgray!40,topline=false,rightline=false,leftline=false,bottomline=false,innertopmargin=2pt]\begin{mdresult}}{\end{mdresult}\end{mdframed}}

\newtheorem{remark}[lemma]{Remark}

\theoremstyle{definition}

\newtheorem{question}{Open Question}

\newtheorem{mdproblem}{Problem}

\newtheorem*{mdproblem*}{Problem}
\newenvironment{Problem*}{\begin{mdframed}[hidealllines=false,innerleftmargin=10pt,backgroundcolor=gray!10,innertopmargin=5pt,innerbottommargin=5pt,roundcorner=10pt]\begin{mdproblem*}}{\end{mdproblem*}\end{mdframed}}
\newtheorem{mddefinition}[lemma]{Definition}
\newenvironment{Definition}{\begin{mdframed}[hidealllines=true,innerleftmargin=10pt,backgroundcolor=ForestGreen!10,innertopmargin=5pt,innerbottommargin=5pt,roundcorner=10pt]\begin{mddefinition}}{\end{mddefinition}\end{mdframed}}
\newtheorem*{mddefinition*}{Definition}
\newenvironment{Definition*}{\begin{mdframed}[hidealllines=false,innerleftmargin=10pt,backgroundcolor=white!10,innertopmargin=5pt,innerbottommargin=5pt,roundcorner=10pt]\begin{mddefinition*}}{\end{mddefinition*}\end{mdframed}}
\newtheorem{mdremark}{Remark}

\newenvironment{ourbox}{\begin{mdframed}[hidealllines=false,innerleftmargin=10pt,backgroundcolor=white!10,innertopmargin=2pt,innerbottommargin=5pt,roundcorner=10pt]}{\end{mdframed}}

\newtheorem{mdalgorithm}{Algorithm}
\newenvironment{Algorithm}{\begin{ourbox}\begin{mdalgorithm}}{\end{mdalgorithm}\end{ourbox}}

\allowdisplaybreaks

\renewcommand{\qed}{\nobreak \ifvmode \relax \else
      \ifdim\lastskip<1.5em \hskip-\lastskip
      \hskip1.5em plus0em minus0.5em \fi \nobreak
      \vrule height0.75em width0.5em depth0.25em\fi}

\newcommand{\Qed}[1]{\rlap{\qed$_{\textnormal{~\Cref{#1}}}$}}

\renewcommand{\leq}{\leqslant}
\renewcommand{\geq}{\geqslant}

\newcommand{\Leq}[1]{\ensuremath{\underset{\textnormal{#1}}\leq}}

\newcommand{\Eq}[1]{\ensuremath{\underset{\textnormal{#1}}=}}

\newcommand{\Ot}{\ensuremath{\widetilde{O}}}
\newcommand{\eps}{\ensuremath{\varepsilon}}
\newcommand{\Paren}[1]{\Big(#1\Big)}
\newcommand{\Bracket}[1]{\Big[#1\Big]}
\newcommand{\bracket}[1]{\left[#1\right]}
\newcommand{\paren}[1]{\ensuremath{\left(#1\right)}\xspace}
\newcommand{\card}[1]{\left\vert{#1}\right\vert}

\newcommand{\IN}{\ensuremath{\mathbb{N}}}

\newcommand{\floor}[1]{{\left\lfloor{#1}\right\rfloor}}

\newcommand{\expect}[1]{\Exp\bracket{#1}}

\newcommand{\set}[1]{\ensuremath{\left\{ #1 \right\}}}
\newcommand{\poly}{\mbox{\rm poly}}
\newcommand{\polylog}{\mbox{\rm  polylog}}

\DeclareMathOperator*{\Exp}{\ensuremath{{\mathbb{E}}}}
\DeclareMathOperator*{\Prob}{\ensuremath{\textnormal{Pr}}}
\renewcommand{\Pr}{\Prob}

\newenvironment{tbox}{\begin{tcolorbox}[
		enlarge top by=5pt,
		enlarge bottom by=5pt,
		 breakable,
		 boxsep=0pt,
                  left=4pt,
                  right=4pt,
                  top=10pt,
                  arc=0pt,
                  boxrule=1pt,toprule=1pt,
                  colback=white
                  ]
	}
{\end{tcolorbox}}

\newcommand{\event}{\ensuremath{\mathcal{E}}}

\newcommand{\II}{\ensuremath{\mathbb{I}}}

\newcommand{\mireal}[1][]{
  \ifx\relax#1\relax%
    \II(\mione \,; \mitwo)%
  \else%
    \II(\mione \,; \mitwo\mid #1)%
  \fi
}

\newcommand{\PP}{\ensuremath{\mathbb{P}}}

\newcommand{\GG}{\ensuremath{\mathcal{G}}}

\title{Faster Vizing and Near-Vizing Edge Coloring Algorithms}
\author{Sepehr Assadi\footnote{(sepehr@assadi.info) Cheriton School of Computer Science, University of Waterloo, and Department of Computer Science, Rutgers University. 
Supported in part by a  Sloan Research Fellowship, an NSERC
Discovery Grant, a University of Waterloo startup grant, and a Faculty of Math Research Chair grant. \smallskip}}

\date{}

\begin{document}
\maketitle

\pagenumbering{roman}

\begin{abstract}

\bigskip

Vizing's celebrated theorem states that every simple graph with maximum degree $\Delta$ admits a $(\Delta+1)$ edge coloring which can  be found in $O(m \cdot n)$ time on $n$-vertex $m$-edge graphs. 
This is just one  color more  than the trivial lower bound of $\Delta$ colors needed in any proper edge coloring. 
After a series of simplifications and variations, this 
running time was eventually improved by Gabow, Nishizeki, Kariv, Leven, and Terada in 1985 to $O(m\sqrt{n\log{n}})$ time. This has effectively remained the state-of-the-art modulo an $O(\sqrt{\log{n}})$-factor improvement 
by Sinnamon in 2019. 

\medskip

As our main result, we present a novel randomized algorithm that computes a $\Delta+O(\log{n})$ coloring of any given simple graph in $O(m\log{\Delta})$ expected time; in other words, a near-linear time\footnote{Throughout, by 
`near-linear' time, we mean the common convention of $\Ot(m) := O(m \cdot \polylog{(n)})$ time.} randomized algorithm for a ``near''-Vizing's coloring. 

\medskip

As a corollary of this algorithm, we also obtain the following results: 
\begin{itemize}
	\item A randomized algorithm for $(\Delta+1)$ edge coloring in $O(n^2\log{n})$ expected time.  This is near-linear in the input size for dense graphs and presents the first polynomial time improvement over 
	the longstanding bounds of Gabow et.al. for Vizing's theorem in almost four decades.
	\item A randomized algorithm for $(1+\eps)  \Delta$ edge coloring in $O(m\log{(1/\eps)})$ expected time for any $\eps =  \omega(\log{n}/\Delta)$. The dependence on $\eps$ 
	exponentially improves upon a series of recent results that obtain algorithms with runtime of  $\Omega(m/\eps)$ for this problem. 
\end{itemize}

\medskip

Our main algorithm does {not} rely on any of the previous standard machinery for edge coloring algorithms in this context such as Vizing's fans or chains, Euler partitions and degree splitting, or Nibble methods. 
Instead, it works by repeatedly finding matchings that we call \emph{``fair''}, which, intuitively speaking, match every maximum degree vertex of the graph except with a small (marginal) probability. This technique may be of 
its own independent interest for other edge coloring problems.

\end{abstract}

\clearpage

\setcounter{tocdepth}{3}
\tableofcontents
\clearpage
\pagenumbering{arabic}
\setcounter{page}{1}


\section{Introduction}\label{sec:intro}

Let $G=(V,E)$ be a simple undirected graph with $n$ vertices, $m$ edges, and maximum degree $\Delta$. 
A \textbf{(proper) $\bm{k}$ edge coloring} of $G$ is an assignment $\varphi : E \rightarrow [k]$ of the $k$ colors to the edges of $G$
so that for any pairs of edges $e,f$ that share a vertex, we have $\varphi(e) \neq \varphi(f)$. 
Edge coloring is a fundamental problem in graph theory with a wide range of applications in computer science. It is also a staple 
in virtually every graph theory textbook or course. In this paper, we focus on edge coloring from an  \emph{algorithmic} point of view.

Every graph requires at least $\Delta$ colors in any proper edge coloring to color edges incident on any maximum-degree vertex. 
At the same time, Vizing's celebrated theorem~\cite{Vizing64} proves that every graph admits a $(\Delta+1)$ edge coloring which can also be found in polynomial time (see~\cite[Appendix A]{StiebitzSTF12}
for an English translation of Vizing's paper). 
Yet, deciding whether or not a given graph requires $\Delta$ or $\Delta+1$ colors is NP-hard~\cite{Holyer81}.  
Consequently, the algorithmic research on this problem has naturally shifted toward obtaining more efficient algorithms for $\Delta+1$ edge coloring, sometimes referred to as \emph{Vizing's coloring}, 
or in some cases, more efficient algorithms for even larger number of colors. 

\paragraph{Vizing's coloring.} Vizing's original proof~\cite{Vizing64} is already constructive and can be turned into an $O(m \cdot n)$ time algorithm; see, e.g.,
the simplifications by~\cite{RaoD92,MisraG92} that explicitly achieve this time bound. The runtime for Vizing's coloring was eventually improved by~\cite{GabowNKLT85}
to an $O(m \cdot \min\set{\Delta \cdot \log{n},\sqrt{n\log{n}}}) = \Ot(m\sqrt{n})$ time algorithm. This has remained the state-of-the-art for this problem for nearly four decades 
modulo a $\Theta(\sqrt{\log{n}})$ factor improvement by~\cite{Sinnamon19} that improved the runtime of essentially the same algorithm to $O(m \cdot \min\set{\Delta \cdot \log{n},\sqrt{n}}) = O(m\sqrt{n})$ instead. 
This leads to one of the most central open questions in algorithmic edge coloring: 
\begin{ourbox}
\begin{question}\label{open:vizing}
	Are there faster algorithms for $(\Delta+1)$ edge coloring of simple graphs compared to the longstanding bounds of $\Ot(m\sqrt{n})$ achieved by~\cite{GabowNKLT85}? 
\end{question}
\end{ourbox}

Significant progress has been made on this question in special cases of bipartite graphs~\cite{ColeOS01}, bounded-degree graphs~\cite{BernshteynD23}, or bounded-arboricity graphs~\cite{BhattacharyaCPS23}; however, 
the original question for arbitrary graphs remains unresolved. 

It is worth mentioning that the $\Ot(m\sqrt{n})$ runtime in~\Cref{open:vizing} coincides with a ``classical'' barrier for other fundamental graph problems including maximum matching~\cite{HopcroftK73,EvenK75,MicaliV80,GabowT91} or negative weight shortest path~\cite{Gabow85,GabowT89,Goldberg95}. The recent wave of improvements for ``flow based'' problems using continuous optimization techniques, e.g., in~\cite{SpielmanT04,DaitchS08,Madry13,LeeS14,CohenMSV17,BrandLNPSSSW20}---culminating 
in the seminal $m^{1+o(1)}$ time algorithm of~\cite{ChenKLPGS22} for min-cost flow---has break through this classical barrier for each of these problems that can be reduced to min-cost flow (e.g., \emph{bipartite} matching
or negative weight shortest path). There has also been a recent push by the community
to obtain simpler and ``combinatorial'' algorithms for these problems including for bipartite matching~\cite{ChuzhoyK24a,ChuzhoyK24b} and negative weight shortest path~\cite{BernsteinNW22,BringmannCF23}. 
Nevertheless, these lines of work have not yet proven useful for ``non-flow based'' problems such as non-bipartite maximum matching, and more relevant to us, $(\Delta+1)$ edge coloring. 

\paragraph{Greedy and approximate coloring.} When allowing a larger number of colors, a $(2\Delta-1)$ edge coloring can be found in $O(m\log{\Delta})$ time using a textbook greedy algorithm: iterate over the edges 
in an arbitrary order and for each edge find a color not used to color edges of any of its endpoints using a binary search (the existence of such a color follows directly from the pigeonhole principle).  

Motivated in part by the research on dynamic edge coloring algorithms, there has been a recent surge of interest in 
extending this algorithm to also achieve $(1+\eps)\Delta$ coloring for $\eps \in (0,1)$ in near-linear time. For instance, \cite{DuanHZ19} gave a randomized $(1+\eps)\Delta$ edge coloring algorithm in $\Ot(m/\eps^2)$ 
time as long as $\Delta = \Omega(\log{(n)}/\eps)$.~\cite{BhattacharyaCPS24} gave another randomized $(1+\eps)\Delta$ edge coloring algorithm in $O(m\log{(1/\eps)}/\eps^2)$ time as long as 
$\Delta$ is at least some $\polylog{(n)}$ (whose exponent depends on $\eps$). Finally,~\cite{ElkinK24} gave a deterministic $(1+\eps)\Delta$ edge coloring algorithm in $O(m\log{(n)}/\eps)$ time and a randomized one with  $O(m/\eps^{18} + m\log{(\eps \cdot \Delta)})$ expected runtime. Yet, all of these results still require $\Omega(m/\eps)$ time, leading to the following natural question: 

\begin{ourbox}
\begin{question}\label{open:approximate}
	Are there $(1+\eps)\Delta$ edge coloring algorithms with $\ll O(m/\eps)$ runtime? 
\end{question}
\end{ourbox}

For some very small values of $\eps$, say, $\eps = 1/\Delta$, the answer to~\Cref{open:approximate} is already known to be \emph{yes} given the $\Ot(m\sqrt{n})$ time algorithm of~\cite{GabowNKLT85} for $(\Delta+1)$ coloring (assuming 
$\Delta$ is also much larger than $\sqrt{n}$ here). Another notable result here is the randomized $\Ot(m)$ expected time algorithm of~\cite{KarloffS87} that achieves a $(\Delta + O(\sqrt{\Delta} \cdot \log{n}))$ coloring. 
Finally, recent work in~\cite{ChristiansenRV23,BhattacharyaCPS23b} also obtained $\Ot(m)$ time algorithms for $\Delta + O(\alpha)$ coloring of graphs with arboricity $\alpha$ (which addresses~\Cref{open:approximate} for small values of $\eps$ 
on bounded arboricity graphs). However, these results still leave~\Cref{open:approximate} unresolved for the most part. 

\subsection{Our Contributions}\label{sec:results}

We address both~\Cref{open:vizing} and~\Cref{open:approximate} in this work. Our main contribution however lies elsewhere: we present a novel randomized algorithm 
that achieves a ``near''-Vizing's coloring in near-linear time. 

\begin{result}[\textbf{Main Result}; Formalized in~\Cref{thm:fast-approx}]\label{res:main}
	There is a randomized algorithm that in $O(m\log{\Delta})$ expected time outputs a $\Delta + O(\log{n})$ edge coloring of any given simple graph. 
\end{result}

\Cref{res:main} shows that in essentially the same time as the greedy algorithm (modulo being a randomized algorithm and the expected runtime guarantee), we can obtain 
a coloring which is ``nearly'' the same as Vizing's coloring, except for an additive $O(\log{n})$ gap. 
This can also be seen as improving the $\Ot(m)$ time randomized algorithm of~\cite{KarloffS87} that uses $\Delta+O(\sqrt{\Delta} \cdot \log{n})$ colors by reducing the number of ``extra'' colors by a factor of $O(\sqrt{\Delta})$.

Our algorithm in~\Cref{res:main} does {not} rely on any of the previous standard machinery for edge coloring algorithms in this context such as 
Vizing's fans or chains (see, e.g.~\cite{Vizing64,RaoD92,MisraG92,BernshteynD23}), Euler partitions and degree splitting (see, e.g.,~\cite{GabowNKLT85,Sinnamon19,ElkinK24}), or Nibble methods (see, e.g.~\cite{BhattacharyaCPS24}). 
Instead, it works by repeatedly finding what we call \emph{``fair matchings''} -- these are matchings in the input graph which match each maximum degree vertex except with (marginal) probability at most $\approx 1/\Delta$. 
The coloring algorithm then works by coloring each of these matchings with a different color and ``peeling them'' from the graph. Combining the ``fairness'' guarantee of these matchings with a detailed probabilistic analysis allows us to argue that 
after removing $\Delta$ of such matchings, the maximum degree of the remaining graph is only $O(\log{n})$ with high probability; these remaining edges now can be colored  using $O(\log{n})$ extra colors (say, using the greedy algorithm). 
We will elaborate more on our techniques in~\Cref{sec:techniques}.

We are not aware of any prior approach (not even necessarily an efficient algorithm) for edge coloring graphs with a similar number of colors (generally, below $\Delta + \Theta(\sqrt{\Delta})$ colors), that does not 
rely on using Vizing's fans and chains (or similar ``graph search'' approaches such as \emph{Kierstead Paths}~\cite{Kierstead84} or \emph{Tashkinov Trees}~\cite{Tashkinov00}; see~\cite{StiebitzSTF12} for more details). 
This is the key reason we find~\Cref{res:main} to be the main contribution of our work (see also~\Cref{rem:local}). 

Going back to~\Cref{open:vizing}, we can use~\Cref{res:main} together with existing tools to significantly improve the longstanding bounds of~\cite{GabowNKLT85} in case of moderately dense graphs. 

\begin{result}[Formalized in~\Cref{thm:fast-vizing}]\label{res:vizing}
	There is a randomized algorithm that in $O(n^2\log{n})$ expected time outputs a $(\Delta + 1)$ edge coloring of any given simple graph. 
\end{result}

\Cref{res:vizing} improves upon the algorithms of~\cite{GabowNKLT85,Sinnamon19} by roughly a $\Theta(\sqrt{n})$ factor on dense graphs and becomes near-linear time in that case. 
In general, this is a polynomial time improvement on the $\Ot(m\sqrt{n})$ runtime for any moderately dense graph with $m = n^{3/2+\Omega(1)}$ edges. 

Obtaining~\Cref{res:vizing} from~\Cref{res:main} is quite straightforward: we first compute a $\Delta+O(\log{n})$ coloring via~\Cref{res:main} which can be seen as a $(\Delta+1)$ edge coloring of all but 
$O(m\log{(n)}/\Delta)$ edges. We then use Vizing's original approach, via Vizing's fans and chains, to color each one of these remaining edges in $O(n)$ time per edge. This gives an algorithm with expected  $\Ot(mn/\Delta)$ runtime
which is always $\Ot(n^2)$ (but can even be faster on graphs which are far from $\Delta$-regular; for instance, this will be $\Ot(n^2 \cdot \alpha/\Delta)$ on graphs with arboricity $\alpha$). 

Finally, with respect to~\Cref{open:approximate}, we can slightly tweak the algorithm in~\Cref{res:main} to also obtain the following result. 

\begin{result}[Formalized in~\Cref{thm:fast-eps-approximate}]\label{res:approximate}
	There is a randomized algorithm that for $\eps = \omega(\log{\!(n)}/\Delta)$ outputs a $(1+\eps)\Delta$ coloring of any given simple graph in $O(m\log{(1/\eps)})$ expected time. 
\end{result}

\Cref{res:approximate} improves upon a series of recent results in~\cite{DuanHZ19,BhattacharyaCPS24,ElkinK24} for $(1+\eps)\Delta$ edge coloring by improving the $\poly{(1/\eps)}$-dependence in the runtime exponentially 
to only an $O(\log{(1/\eps)})$ dependence. We note that having some lower bound on $\eps$ is quite common in this line of work (see, e.g.~\cite{DuanHZ19,BhattacharyaCPS24}) and additionally, our 
the lower bound on $\eps$ in~\Cref{res:approximate} is rather mild as $\eps \geq 1/\Delta$ always (otherwise, we are asking for $\Delta$-coloring which is NP-hard).

\subsection{Our Techniques}\label{sec:techniques}

We focus solely on our main algorithm in~\Cref{res:main} here, as the rest follow more or less immediately from it. 
This algorithm does not rely on the previous machinery for $(\Delta+1)$ edge coloring in~\cite{Vizing64,RaoD92,MisraG92,GabowNKLT85,Sinnamon19,BhattacharyaCPS23}. 
Instead, it is more in the spirit of prior algorithms for $\Delta$-coloring of \emph{bipartite} (multi-)graphs in~\cite{ColeH82,ColeOS01,Alon03,GoelKK10} and a reduction from them to $(\Delta+O(\sqrt{\Delta}\log{n}))$ coloring of general graphs in~\cite{KarloffS87}. 
In the following, we will briefly discuss this latter line of work and then describe our new algorithm in comparison to them. 
The interested reader is referred to~\cite{Sinnamon19} for a nice summary and simplification of the by-now standard machinery for $(\Delta+1)$ edge coloring in~\cite{Vizing64,RaoD92,MisraG92,GabowNKLT85}. 

\paragraph{$\Delta$-coloring bipartite multi-graphs.} When working with bipartite multi-graphs, without loss of generality we can assume the input is $\Delta$-regular: repeatedly 
contract vertices of degree $<\Delta/2$ on each side of the bipartition 
and then add arbitrary edges between degree $<\Delta$ vertices to turn the graph regular. This makes the graph $\Delta$-regular while keeping the number of multi-edges in the graph at $O(m)$ still (see, e.g.~\cite{Alon03}, for this standard reduction). 

Under this assumption, there is a very simple $O(m\log{\Delta})$ time algorithm for $\Delta$ edge coloring of bipartite multi-graphs as long as $\Delta$ is a power of two. Find a 
Eulerian tour of $G$ in $O(m)$ time (using a standard DFS) and pick alternating edges to form two $\Delta/2$-regular subgraphs. Recurse on each subgraph
for $\log{\Delta}$ steps to partition all edges of $G$ into $\Delta$ separate $1$-regular subgraphs or alternatively perfect matchings. Then, color each of these perfect matchings with a different color
to obtain a $\Delta$-coloring of $G$. This idea, with quite some more technical work, was extended to all values of $\Delta$ in~\cite{ColeH82} for non-sparse graphs and~\cite{ColeOS01} for all graphs; see also~\cite{Alon03} for a much simpler and a direct analogue of this approach
for all $\Delta$, but with a slightly worse runtime of $O(m\log{n})$ instead. A different implementation of this idea is due to~\cite{GoelKK10} who designed a \emph{random walk} approach for finding 
a perfect matching in \emph{regular} bipartite multi-graphs in $O(n\log{n})$ expected time (independent of $m$)---one can then repeatedly apply their algorithm to find $\Delta$ perfect matchings in $G$ and color them differently
which leads to an $O(m\log{n})$ expected time algorithm. 

\paragraph{Reducing edge coloring of general graphs to bipartite ones.} The authors in \cite{KarloffS87} observed that one can reduce edge coloring of general graphs to the easier bipartite case, albeit with some loss. 
The idea is as follows: partition the vertices of $G$ randomly into two sides $L$ and $R$ and let $H$ be the bipartite subgraph of $G$ on this bipartition. Using the randomization of this partitioning, one can prove
that both $H$ and $G \setminus H$ have maximum degree $\Delta/2 + O(\sqrt{\Delta}\cdot\log{n})$ with high probability. But now, we can color $H$ using these many colors in $\Ot(m)$ time
and recurse on $G \setminus H$ to color the remaining edges, using a disjoint set of colors for $H$ and $G \setminus H$. One can then prove that this algorithm takes only $\Ot(m)$ time in total and produces 
a $\Delta + O(\sqrt{\Delta}\cdot\log{n})$ coloring of the entire graph $G$. 

\subsubsection*{Our Approach}
 At a conceptual level, our approach can be seen as a \emph{``local''} reduction of general graphs to bipartite graphs that allows 
us to avoid the rather heavy loss in the number of extra colors in~\cite{KarloffS87}. Such a local reduction itself is inspired by the classical algorithms for non-bipartite matching---dating all the way back to the seminal work of~\cite{Edmonds65}---that use \emph{Blossoms} to make the graph ``locally bipartite'' when finding augmenting paths for a given matching. Given that the previous work on $\Delta$-coloring bipartite multi-graphs, specifically the random walk algorithm of~\cite{GoelKK10}, 
are heavily based on finding matchings, can we find a ``more local'' reduction than~\cite{KarloffS87}? 

Let us point out three obstacles in extending the approach of~\cite{GoelKK10} to non-bipartite graphs: 
\vspace{-20pt}
\begin{enumerate}
\item[$(a)$] Unlike in bipartite graphs, we cannot assume the graph is $\Delta$-regular: any such reduction inherently needs to 
create parallel edges but unlike simple graphs, multi-graphs may \emph{not} even admit a $\Delta+1$ edge coloring (only a $\floor{3\Delta/2}$ coloring; see the so-called \emph{Shannon} multi-graphs)\footnote{Note that one can 
still make the graph $\Delta$-regular by increasing the number of edges to $O(n\Delta)$ instead; see, e.g.~\cite{KahnK23}. This already increases the runtime beyond $\Ot(m)$ but more importantly, it will not be particularly helpful given the second obstacle.};
\item[$(b)$] Even if we start with a $\Delta$-regular input graph $G$, there may not even be a perfect matching inside $G$. Thus, we cannot remove a perfect matching and reduce the problem 
to that of coloring a $(\Delta-1)$-regular graph as is the case in the edge-coloring algorithm of~\cite{GoelKK10}; 
\item[$(c)$] And finally, the random walk approach of~\cite{GoelKK10} only finds a perfect matching in \emph{bipartite} graphs to begin with and cannot be applied to a non-bipartite graph. 
\end{enumerate}

We bypass these three obstacles by introducing what we call \textbf{{fair} matchings}: we say that a \emph{distribution} over matchings $M$ in a graph $G$ 
outputs a fair matching iff for every maximum degree vertex $v$ in $G$, the probability that $v$ is \emph{not} matched by $M$ is at most $1/\Delta$ (plus some negligible $1/\poly(\Delta)$ additive factor). We shall emphasize
that this probabilistic guarantee only holds \emph{marginally}, namely, the choice of which vertices may remain unmatched can be (heavily) correlated. 
The existence of such a distribution in every graph is a corollary of Vizing's theorem itself: simply pick one of the color classes uniformly at random. But of course, for our purpose, finding a fair matching this way
will be counter productive. 

As the first main part of our approach, we show that, modulo a one-time \emph{preprocessing} step, we can sample a fair matching in every graph of maximum degree $\Delta$ much faster and only in $O(n_{\Delta} \cdot \log{\Delta})$ time, 
where $n_{\Delta}$ denotes the number of maximum-degree vertices in $G$. This algorithm itself is inspired by the random walk approach of~\cite{GoelKK10} but with several modifications needed to handle 
non-regular\footnote{This step requires working with a considerably more complicated Markov Chain compared to~\cite{GoelKK10} for the analysis of the random walk process.} and non-bipartite\footnote{This is 
where we use the ``local reduction'' to bipartite graph idea, similar to non-bipartite matching algorithms in~\cite{Edmonds65,GabowT91}, although without any Blossom contraction and instead using
a simple labeling scheme (somewhat similar to~\cite{Balinski67}) and the randomness of the walk itself.} graphs and (more importantly) ensuring the ``fairness'' guarantee of the output matching\footnote{In particular, in sharp contrast with~\cite{GoelKK10} that always grows a matching $M$, 
our algorithm spends most of its time in \emph{unmatching} and \emph{rematching} 
the vertices in the output matching to ensure it ``converges'' to a fair matching.}.

The second main part is how to use these fair matchings to obtain a $(\Delta+O(\log{n}))$ edge coloring of the input graph $G$. The algorithm works as follows: it first finds a fair matching $M_\Delta$ from $G$ (for degree $\Delta$ vertices) and 
remove it from the graph. Then, for every maximum-degree vertex not matched by $M_\Delta$, it picks a random incident edge and let $F_\Delta$ be these edges and additionally remove $F_\Delta$ from $G$. 
This way, after spending $O(n_\Delta \cdot \log{\Delta})$ time, the algorithm obtains $M_{\Delta}$, $F_{\Delta}$, and a new graph with maximum degree at most $\Delta-1$. The algorithm then recurse on the remaining graph (and new max-degree) 
to obtain $M_{\Delta-1},F_{\Delta-1}$, and so on until obtaining $M_1,F_1$. We can show that this algorithm only takes $O(m\log{\Delta})$ time 
up until here. Moreover, since $M_{\Delta},\ldots,M_1$ are matchings, the algorithm can color each of them using a different color and spends $\Delta$ colors. Finally, 
we can also prove that the subgraph $F := F_{\Delta} \cup \ldots  \cup F_1$ has a maximum degree $O(\log{n})$ with high probability\footnote{By ``fairness'' guarantee,  
the expected number of times a vertex remains unmatched is at most $\sum_{i=1}^{\Delta} 1/i \leq \ln(\Delta)$. The challenge however is to handle the edges picked to a vertex by its unmatched neighbors, 
given that the choice of those neighbors may even be correlated (as the guarantee of fair matchings is only on marginal probabilities).}. Hence, we can simply run a greedy algorithm in the final step to color these remaining edges
with $O(\log{n})$ new colors, and obtain a $(\Delta+O(\log{n}))$ edge coloring of $G$.


\newcommand{\Gnew}{G_{\textsf{new}}}
\newcommand{\Enew}{E_{\textsf{new}}}
\newcommand{\Vnew}{V_{\textsf{new}}}

\newcommand{\degnew}{\deg_{\textsf{new}}}
\newcommand{\degrem}{\deg_{\textnormal{\textsf{rem}}}}
\newcommand{\Nnew}{N_{\textsf{new}}}

\newcommand{\indeg}{\deg^{-}}
\newcommand{\outdeg}{\deg^{+}}

\newcommand{\mumin}{\mu_{\textsf{min}}}

\newcommand{\mumax}{\mu_{\textsf{max}}}

\section{Preliminaries} 

\paragraph{Notation.} For any integer $t \geq 1$, we use $[t] := \set{1,2,\ldots,t}$. 
For a graph $G=(V,E)$, we denote the vertices by $V(G) := V$ and edges by $E(G) := E$ and their sizes by $n(G) := \card{V}$ and $m(G) := \card{E}$; if clear from the context, we simply refer to them as $n$ and $m$, respectively. For any set of edges $F \subseteq E$, $V(F)$ denotes the vertices incident on $F$.  

For an undirected graph $G$, $\Delta(G)$ is the maximum degree and $\deg_G(v)$ is the degree of $v$ in $G$ for each $v \in V$; we also
use $N_G(v)$ to denote the neighbor-set of $v$. For a directed graph (digraph) $G=(V,E)$, we use $\indeg_G(v)$ and $\outdeg_G(v)$ to denote the in-degree and out-degree of $v \in V$, and $N^-_G(v)$ and $N^+_G(v)$ as the in-neighbors and out-neighbors of $v$. We may drop the subscript $G$ from the above notation whenever it is clear from the context. We say a digraph $G$ is \textbf{balanced} iff we have $\indeg(v) = \outdeg(v)$ for every $v \in V$. 

For a matching $M \subseteq E$, and any vertex $v \in V$, $M(v)$ denotes the matched neighbor of $v$. If $v$ is unmatched, then $M(v) = \emptyset$. 
An alternating path $P$ for a matching $M$ is a path consisting of edges alternating between $E \setminus M$ and $M$. For a matching $M$, \emph{applying} the alternating path $P$ to 
$M$ results in another matching $M' := M \triangle P$, namely, the symmetric differences of $M$ and $P$. 

\subsection{Probabilistic Inequalities}
 We use the following standard version of Chernoff bound. 

\begin{proposition}[\textbf{Chernoff bound}; c.f.~\cite{AlonS16}]\label{prop:chernoff}
	Let $X_1,\ldots,X_n$ be $n$ independent random variables in $[0,b]$ each. Define $X := \sum_{i=1}^{n} X_i$. For any $\delta \in (0,1]$ and $\mumin \leq \expect{X} \leq \mumax$, 
	\[
		\Pr\paren{X \geq (1+\delta) \cdot \mumax} \leq \exp\paren{-\frac{\delta^2 \cdot \mumax}{3b}} \quad \text{and} \quad \Pr\paren{X \leq (1-\delta) \cdot \mumin} \leq \exp\paren{-\frac{\delta^2 \cdot \mumin}{3b}}.
	\]
\end{proposition}

We also prove the following concentration inequality for sum of independent random variables with unbounded range but admitting an exponentially decaying tail. 

\begin{proposition}\label{prop:unbounded-chernoff}
	Let $\mu_1,\ldots,\mu_n$ be $n$ real non-negative numbers and $X_1,\ldots,X_n$ be $n$ independent non-negative random variables such that for every $i \in [n]$:
	\[
		\expect{X_i} \leq \mu_i \quad \text{and} \quad \forall x \geq 1: ~ \Pr\paren{X_i \geq x} \leq \eta \cdot e^{-\kappa \cdot x} \cdot \mu_i, 
	\]
	for some $\eta, \kappa >0$. Define $X := \sum_{i=1}^{n} X_i$ and let $\mumax \geq \sum_{i=1}^{n} \mu_i$. Then, 
	\[
		\Pr\paren{X \geq \frac{6\eta}{\kappa^2} \cdot \mumax} \leq \exp\paren{-\frac{\eta}{\kappa} \cdot \mumax}. 
	\]
\end{proposition}

The proof is standard and appears in~\Cref{app:concentration} (plus some further remarks about this bound). 

Finally, we also use the following auxiliary inequality in our probabilistic analysis. 

\begin{fact}\label{prop:weird-inequality}
	For any set of $n$ numbers $x_1,\ldots,x_n$ in $[0,1]$,
	\[
		\sum_{i=1}^{n} x_i \cdot \prod_{j=i+1}^{n} (1-x_j) \leq 1; 
	\]
\end{fact}
\begin{proof}
	We prove this by induction on $n$. The base of $n=1$ trivially holds as the expression in the LHS is $x_1$ and we have $x_1 \in [0,1]$. 
	For larger $n$, we have, 
	\begin{align*}
		\sum_{i=1}^{n} x_i \cdot \prod_{j=i+1}^{n} (1-x_j) &= (1-x_n) \cdot \paren{\sum_{i=1}^{n-1} x_i \cdot \prod_{j=i+1}^{n-1} (1-x_j)} + x_n \leq (1-x_n) \cdot 1 + x_n  = 1. \qed
	\end{align*}
	
\end{proof}

\subsection{Basic Facts about Random Walks}

We use some standard facts about random walks on \emph{directed} graphs. By a \textbf{standard random walk} on a directed graph $G=(V,E)$, we mean a walk that at every step picks one of the out-neighbors of the current 
vertex uniformly at random to visit next. 

The \textbf{hitting time} of a vertex $u$ to vertex $v$, denoted by $h_{u,v}$, is defined as the expected length of the walk starting at $u$ to visit $v$ for the first time.
The \textbf{return time} of a vertex $v$, denoted by $h_{v,v}$, is the expected length of the walk starting at $v$ before it returns to $v$ again. 
Throughout, we use $\pi$ to denote the \textbf{stationary distribution} of the random walk. 

The proofs of the following fact can be found in~\cite[Chapter~2.2]{AldousF02}. 

\begin{fact}\label{fact:random-walk}
In a standard random walk on a directed strongly connected (multi-) graph $G=(V,E)$: 
\begin{enumerate}
	\item\label{rw-fact:ftmc} For any vertex $v \in V$, $h_{v,v} = 1/\pi_v$; 
	\item\label{rw-fact:hit-before-return} Let $W_v$ be a random walk that starts from $v$ until it returns to it. For any $u \in V$, 
		\[
			\expect{\text{\# number of times $u$ appears in $W_v$}} = \frac{\pi_u}{\pi_v}. 
		\]
	\item\label{rw-fact:balanced} If $G$ is \underline{balanced}, then $\pi_v = \outdeg(v)/m(G)$ for every $v \in V$. 
\end{enumerate}
\end{fact}

\clearpage


\newcommand{\MM}{\ensuremath{\mathcal{M}}}
\newcommand{\VV}{\ensuremath{\mathcal{V}}}

\newcommand{\WW}{\ensuremath{\mathcal{W}}}

\newcommand{\algname}[1]{\ensuremath{\textnormal{\texttt{#1}}}}

\newcommand{\VVstart}{\ensuremath{\VV_{\textnormal{\textsf{s}}}}}

\newcommand{\VVmatch}{\ensuremath{\VV_{\textnormal{\textsf{M}}}}}
\newcommand{\VVend}{\ensuremath{\VV_{\textnormal{\textsf{t}}}}}

\renewcommand{\PP}{\ensuremath{\mathcal{P}}}

\renewcommand{\GG}{\ensuremath{\mathcal{G}}}

\newcommand{\EE}{\ensuremath{\mathcal{E}}}

\newcommand{\VVmatchs}{\ensuremath{\VV_{\textnormal{\textsf{M,stop}}}}}

\newcommand{\VVmatchc}{\ensuremath{\VV_{\textnormal{\textsf{M,cont}}}}}

\newcommand{\rwm}{\algname{FairMatching}\xspace}

\newcommand{\degDelta}{\deg_{\Delta}}

\newcommand{\extract}{\algname{AltPath}\xspace}

\newcommand{\fixeven}{\algname{FixEven}\xspace}

\newcommand{\fixodd}{\algname{FixOdd}\xspace}

\newcommand{\unmatch}[1]{\ensuremath{\textnormal{\textsc{unmatch}}(#1)}\xspace}

\newcommand{\unmatchDelta}[1]{\ensuremath{\textnormal{\textsc{unmatch}}_{\Delta}(#1)}\xspace}

\newcommand{\Gstar}{\ensuremath{G^*}}

\newcommand{\Vstar}{\ensuremath{V^*}}

\newcommand{\Mstar}{\ensuremath{M^*}}

\newcommand{\Nstar}{\ensuremath{N^*}}

\section{A Fast Algorithm for Fair Matchings} 

We present a key subroutine for our fast coloring algorithms in this section: an algorithm that can find a ``fair matching'', namely, 
a matching incident on maximum-degree vertices in a graph, which leaves each such vertex unmatched only with a small probability.
This algorithm is inspired by the random walk algorithm of~\cite{GoelKK10} for perfect matching on \emph{regular bipartite} graphs (although neither regularity nor bipartiteness necessarily hold for us).

\begin{theorem}\label{thm:random-walk-matching}
	There is a randomized algorithm $\rwm$ with the following properties. 
	Let $G=(V,E)$ be a simple undirected graph with maximum degree $\Delta$. Let $V_{\Delta}$ denote the vertices with degree $\Delta$ in $G$ and $n_{\Delta} := \card{V_{\Delta}}$. 
	Then, $\rwm(V,E,\Delta,V_{\Delta})$ outputs a matching $M \subseteq E$ such that for every vertex $v \in V_{\Delta}$: 
	\begin{align}
		\Pr\paren{\text{$v$ is unmatched by $M$}} \leq \frac{1}{\Delta} + \frac{2}{\Delta^2}. \label{eq:has-exponent}
	\end{align}
	$\rwm$ assumes the following access to $G$: 
	\begin{itemize}
		\item Given a vertex $u \in V_\Delta$, output a vertex $v \in N(u)$ chosen uniformly at random; 
		\item Given two vertices $u,v \in V_\Delta$ output if $(u,v)$ is an edge in $E$ or not. 
	\end{itemize}
	Finally,  \rwm runs in $O(n_{\Delta} \cdot \ln{\Delta})$ time in expectation. (The $2/\Delta^2$ term in~\Cref{eq:has-exponent} is not sacrosanct and can be made any $1/\poly(\Delta)$ within the same asymptotic runtime). 
\end{theorem}

The rest of this section is organized as follows. We first define a random walk process that allows for updating an already-computed matching. We then design an algorithm that
can ``recover'' an alternating path from such a random walk and examine its properties. Finally, we use these results to obtain our final algorithm and conclude the proof of~\Cref{thm:random-walk-matching}.   
To keep the flow of arguments, we postpone the implementation details and runtime of all intermediate algorithms to~\Cref{sec:implementation}.

Throughout this entire section, we fix a choice of $(V,E,\Delta,V_{\Delta})$ as the underlying graph and all references will be with respect to them unless
explicitly stated otherwise.

\subsection{The Matching Random Walk}\label{sec:matching-rw}

We start by defining our random walk process. This is different from the process in~\cite{GoelKK10}
and in particular its goal is to find \emph{alternating} paths and not necessarily \emph{augmenting} paths.  For our purpose this is both sufficient and necessary as shall become clear later. 

\begin{Definition}\label{def:matching-walk}
	Let $(V,E,\Delta,V_{\Delta})$ be as in~\Cref{thm:random-walk-matching} and $M \subseteq E$ be a matching. We define a \textbf{Matching Random Walk} with respect to $M$ as the following stochastic process: 
	\begin{enumerate}
		\item Sample an \underline{unmatched} vertex $v_1$ in $V_\Delta$ uniformly at random. 
		\item Starting with $i=1$, run the following loop:
		\begin{enumerate}
			\item Sample $v_{2i}$ from $N(v_{2i-1})$ uniformly at random;
			\item If $v_{2i}$ is \underline{unmatched}, return the walk $W := v_1,\ldots,v_{2i}$ and \textbf{terminate}. 
			\item Otherwise, let $v_{2i+1} = M(v_{2i})$; if $v_{2i+1}$ is \underline{not} in $V_{\Delta}$, return the walk $W:= v_1,\ldots,v_{2i+1}$ and \textbf{terminate}; otherwise, continue the loop. 
		\end{enumerate}
	\end{enumerate}
\end{Definition} 

In words, a Matching Random Walk starts from a random max-degree vertex and pick alternate edges from $E$ (randomly out of the current vertex) and $M$ (matched pair of the current vertex), until it 
either hits an unmatched vertex at any position of the walk or a non-max-degree vertex at an odd position. 

We should also emphasize that unlike~\cite{GoelKK10}, the Matching Random Walk is \emph{not} necessarily even an alternating walk -- it may traverse two edges of $M$ in a row or create odd cycles; in fact, a bulk of our effort 
will be spent on making sure one can \emph{extract} an alternating path from such a walk (with a sufficiently large probability). But that will be handled in the subsequent subsection and for now, we will focus on proving 
the main properties of the walk.

\paragraph{Notation.} Let $M \subseteq E$ be a matching and let $W = v_1,v_2,v_3,\ldots$ be a Matching Random Walk with respect to $M$. We define the \textbf{odd vertices} of $W$ as $odd(W) := v_1,v_3,\ldots,v_{2i-1},\ldots,$ and 
 its \textbf{even vertices} as $even(W) := v_2,v_4,\ldots,v_{2i},\ldots$, respectively. 
 We define $\unmatchDelta{M}$ as the number of vertices in $V_{\Delta}$ that are \underline{not} matched by $M$. 
 
 Additionally, we say $M$ is a \textbf{covered} matching if $V_{\Delta}$ is a vertex cover for it, i.e., at least one endpoint of each of its edges belong to $V_{\Delta}$. 
 Throughout this subsection, our focus will be entirely on covered matchings. 
 
 The following lemma lists the main properties of this random walk process that we will use.  

\begin{lemma}\label{lem:rw-properties}
	Let $M \subseteq E$ be a covered matching with $\unmatchDelta{M} > 0$ and $W=v_1,v_2,\ldots,$ be a Matching Random Walk with respect to $M$. Then: 
	\begin{enumerate}
	\item\label{part:rw-properties-1} The expected length of $W$ is $\Exp\card{W} \leq 7n_{\Delta}/\unmatchDelta{M}$. 
	\item\label{part:rw-properties-2} For any \underline{unmatched} vertex $v \in V_{\Delta}$: $\Pr\paren{v = v_1} = 1/\unmatchDelta{M}$. 
	\item\label{part:rw-properties-3} For any \underline{matched} vertex $v \in V_{\Delta}$: $\expect{\text{$\#$ of times $v$ appears in $odd(W)$}} \leq 1/\unmatchDelta{M}$. 
	\end{enumerate}
\end{lemma}

We prove~\Cref{lem:rw-properties} in the rest of this subsection. 

\paragraph{Notation.} We need yet another round of notation. 
Let $\Gstar$  be a subgraph of $G$ that consists of only vertices that can be part of {some} Matching Random Walk, which we denote by $\Vstar$. 
Define 
\[
\Vstar_{\Delta} = V_{\Delta} \cap \Vstar \quad \text{and} \quad \Mstar = M \cap \Gstar.
\]
\vspace{-10pt}
\begin{observation}\label{obs:Vstar}
Any edge of $M$ is either entirely in $\Gstar$ or entirely in $G \setminus \Gstar$. Moreover, vertices in $\Vstar_\Delta$ have no neighbors in $G \setminus \Gstar $ even in $G$; thus, neighbors of the \emph{remaining} max-degree vertices 
are the same in $\Gstar$ and $G$ and Matching Random Walks are equivalent between the two graphs.  
\end{observation}
Given this observation, for the rest of this subsection, we focus solely on $\Gstar$ and $\Mstar$ which are the only parts relevant for Matching Random Walks\footnote{One can show that the only vertices in $G \setminus \Gstar$ are either (a) 
unmatched by $M$ and are not in $V_{\Delta}$ nor have any neighbors in $V_{\Delta}$, or (b) matched by $M$ and their only neighbors in $V_{\Delta}$ can be their matched pair in $M$. However, we will not need this characterization and hence omit the (rather uninformative) proof.}.

Additionally, for every vertex $v \in V$, we define 
\[
N_{\Delta}(v) := N_{\Gstar}(v) \cap \Vstar_{\Delta} \quad \text{and} \quad \degDelta(v) := \card{N_{\Delta}(v)}
\]
i.e., the neighbors of $v$ in $\Gstar$ whose degree is $\Delta$ and the number of these vertices, respectively.

To prove~\Cref{lem:rw-properties}, we model a Matching Random Walk via a standard random walk on the following digraph (see~\Cref{fig:M-digraph} for an illustration and the discussion after it): 

\begin{Definition}\label{def:matching-mc}
	For any covered matching $M \subseteq E$, we define the \textbf{$\bm{M}$-digraph} as the following directed multi-graph $\GG=(\VV,\EE)$:
	\begin{itemize}
		\item The vertices in $\VV$ are partitioned into four types:
		\begin{alignat*}{2}
			&\set{s,t} && \tag{designated `start' and `target' vertices} \\
			&\VVstart &&:= \set{\alpha(v) \mid v \in \Vstar_\Delta \setminus \Vstar(M)}, \tag{a vertex per each unmatched vertex in $\Vstar_\Delta$} \\
			&\VVmatch &&:= \set{\beta(u,v)~\text{and}~\beta(v,u) \mid \text{$(u,v) \in \Mstar$}}, \tag{a vertex per \emph{direction} of an edge in $\Mstar$} \\
			&\VVend &&:= \set{\gamma(v) \mid v \in \Vstar \setminus \Vstar(M)} \tag{a vertex per each unmatched vertex in $\Vstar$}. 
		\end{alignat*}
		The vertices in $\VVmatch$ are further partitioned into two types: 
		\begin{alignat*}{2}
			\VVmatchc := \set{\beta(u, v) \in \VVmatch \mid v \in \Vstar_\Delta}, \quad 
			\VVmatchs := \set{\beta(u, v) \in \VVmatch \mid v \in \Vstar\setminus \Vstar_\Delta}. 
		\end{alignat*} 
		\item The edges in $\EE$ are as follows: 
		\begin{itemize}[leftmargin=5pt]
			\item $s$ has $\Delta$ parallel directed edges to each $\alpha(v) \in \VVstart$; 
			\item each $\alpha(v) \in \VVstart$ has a directed edge to $\beta(u,w) \in \VVmatch$ or $\gamma(u) \in \VVend$, for $u \in N_{\Gstar}(v)$; 
			\item each $\beta(u, v) \in \VVmatchc$, has a directed edge to $\beta(w, z) \in \VVmatch$ or $\gamma(w) \in \VVend$, for $w \in N_{\Gstar}(v)$;
			 \item each $\beta(u, v) \in \VVmatchs$ has $\Delta$ parallel directed edges to $t$; 
			 \item each $\gamma(v) \in \VVend$ has $\degDelta(v)$ parallel directed edges to $t$; 
			 \item $t$ has $\Delta-\degDelta(u)$ parallel directed edges to each $\beta(u,v) \in \VVmatch$ and another $\card{\VVstart} \cdot \Delta$ to $s$. 
		\end{itemize}
	\end{itemize}
\end{Definition}

\begin{figure}[t!]
	\centering
	\includegraphics[width=0.95\linewidth]{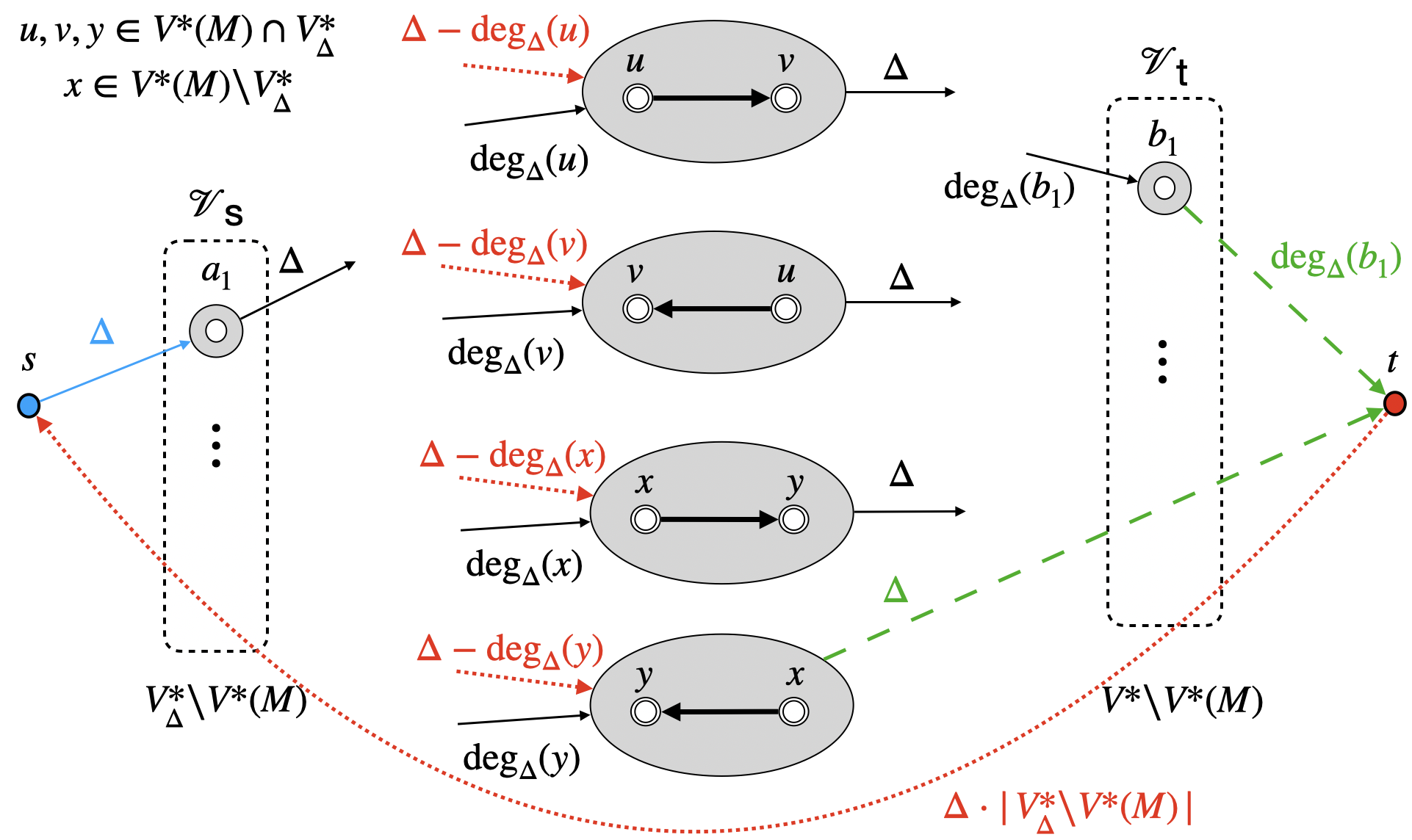}
	\caption{An illustration of the $M$-digraph $\GG$. The vertices of $\GG$ are shown via (gray) circles encompassing a vertex (white circle) or a matching edge (thick black) from the original graph. 
	The vertices $s$ and $t$ are designated vertices unique to $\GG$. The (black) thin edges in the middle are connected according to the original graph, the edges going out of $s$ (blue) are only going to $\VVstart$, 
	the long-dashed edges (green) are going to $t$, and short-dashed edges (red) are coming from $t$. Finally, the middle edges (black) are simple edges and the number next to them shows the different number 
	of those edges leaving or arriving at a vertex, while all other edges are parallel edges and their number shows their multiplicity. 
	}
	\label{fig:M-digraph}
\end{figure}

Informally speaking, the $M$-digraph consists of two designated vertices $s$ and $t$, two vertex per each vertex in $\Vstar_{\Delta}$ unmatched by $\Mstar$ (in $\VVstart$ and $\VVend$), one vertex per each vertex in $\Vstar \setminus \Vstar_{\Delta}$ unmatched by $\Mstar$
(in $\VVend$), and two vertices for each edge in $\Mstar$, one for each direction of traversing the edges (in $\VVmatch$)\footnote{By the definition of the walk in~\Cref{def:matching-walk}, if one direction of a matching edge can appear 
in some walk, the other direction can also appear in some walk (e.g., by traversing this edge itself in the opposite direction).}. The unmatched
edges of $G$ also appear in the $M$-digraph in a way that allows us to move according to $G$ inside the $M$-digraph as well. Finally, the extra edges incident on $s$ and $t$ are added in a careful way to ensure 
that $M$-digraph is balanced (see~\Cref{lem:balanced}). 

We can relate a Matching Random Walk to a standard random walk on $M$-digraphs. 

\begin{observation}\label{obs:mc-upper-bound}
	Suppose we perform a standard random walk in an $M$-digraph $\GG$ by starting from the vertex $s$ and taking an outgoing edge chosen uniformly at random at each vertex until it reaches the vertex $t$ for the first time. 
	Then, the sequence of vertices in $V$ that are visited in this random walk is distributed as a Matching Random Walk with respect to $M$. 
\end{observation}
\begin{proof}
	The proof  is by the definition of the $M$-digraph and the fact that the edges going out of $t$ to other vertices will never be traversed in this standard random walk from $s$ to $t$. 
\end{proof}

\Cref{obs:mc-upper-bound} now allow us to focus on $M$-digraphs instead. The following key lemma shows that $M$-digraphs are balanced and thus we can compute the 
stationary distribution of the random walk over them easily (using~\Cref{fact:random-walk}-\eqref{rw-fact:balanced}, although we need to also check strong connectivity). 

\begin{lemma}\label{lem:balanced}
	For any covered matching $M \subseteq E$, the $M$-digraph $\GG$ is \textbf{\emph{balanced}}, i.e., every vertex has the same in-degree and out-degree. 
\end{lemma}
\begin{proof}
	We go over different types of vertices one by one to verify that their in-degree is equal to their out-degree (although we note that effectively, the only non-trivial case is for $t$; see~\Cref{fig:M-digraph}). 
	\begin{itemize}
	\item \textbf{$\bm{s}$} has $\Delta$ parallel edges to each vertex in $\VVstart$ and in turn receives $\card{\VVstart} \cdot \Delta$ incoming edges from $t$, thus $\indeg_{\GG}(s) = \outdeg_{\GG}(s) = \card{\Vstar_{\Delta} \setminus \Vstar(M)} \cdot \Delta$ ($=\unmatchDelta{\Mstar} \cdot \Delta$).  
	
	\item \textbf{$\bm{\alpha(v) \in \VVstart}$} receives $\Delta$ edges from $s$ and has one outgoing edge per $u \in N_{\Gstar}(v)$ (to either $\beta(u,v)$ or $\gamma(u)$); since $v \in \Vstar_{\Delta}$, $\deg_{\Gstar}(v) = \Delta$ (by~\Cref{obs:Vstar})
	and thus $\alpha(v)$ has $\Delta$ outgoing edges. 
	
	\item \textbf{$\bm{\beta(u,v) \in \VVmatchc}$} receives one incoming edge per $w \in N_{\Delta}(u)$ (from either $\beta(\cdot,w)$ or $\alpha(w)$) and thus has $\degDelta(u)$ incoming edges that way. It 
	also receives $\Delta-\degDelta(u)$ edges from $t$, making its in-degree equal to $\Delta$. On the other hand, $\beta(u,v)$ sends one outgoing edge per $w \in N_{\Gstar}(v)$ (to either $\beta(w,\cdot)$ or $\gamma(w)$) 
	since $v \in \Vstar_{\Delta}$ and thus $\beta(u,v)$ sends $\Delta$ edges out. 
	
	\item \textbf{$\bm{\beta(u,v) \in \VVmatchs}$} receives $\Delta$ incoming edges exactly as in the previous case (choice of $u$ is irrelevant to whether $\beta(u,v)$ is in $\VVmatchc$ or $\VVmatchs$). 
	It also sends $\Delta$ directed edges directly to $t$ and has no other outgoing edges. 
	
	\item \textbf{$\bm{\gamma(v) \in \VVend}$} receives $\degDelta(v)$ edges from vertices in $\VVmatchc$ and $\VVstart$ in total (exactly as in the previous two cases), and sends $\degDelta(v)$ edges
	to $t$ also. 
	
	\item \textbf{$\bm{t}$} receives $\degDelta(v)$ for each $\gamma(v) \in \VVend$ and $\Delta$ edges per each vertex in $\VVmatchs$. Thus, 
	\[
		\indeg_{\GG}(t) = \hspace{-10pt}{\sum_{v \in \Vstar \setminus \Vstar(M)} \hspace{-10pt} \degDelta(v)} + \card{\Vstar(M) \setminus \Vstar_{\Delta}} \cdot \Delta.
	\]
	Similarly, $t$ has $\Delta-\degDelta(u)$ edges to each $\beta(u,v) \in \VVmatch$ and $\card{\Vstar_{\Delta} \setminus \Vstar(M)} \cdot \Delta$ edges to $s$. Thus, 
	\begin{align*}
		\outdeg_{\GG}(t) &= \hspace{-5pt} \sum_{v \in \Vstar(M)}  \hspace{-5pt}(\Delta-\degDelta(v)) + {\card{\Vstar_\Delta \setminus \Vstar(M)} \cdot \Delta} 
		= \paren{\card{\Vstar_\Delta} + \card{\Vstar(M) \setminus \Vstar_{\Delta}}} \cdot \Delta - \hspace{-5pt}\sum_{v \in \Vstar(M)} \hspace{-5pt}\degDelta(v), 
	\end{align*}
	where the second equality is just by re-organizing the terms. Finally, we have, 
	\begin{align*}
		\indeg_{\GG}(t) - \outdeg_{\GG}(t) &=  \sum_{v \in \Vstar} \degDelta(v) - \card{\Vstar_\Delta} \cdot \Delta = 0,
	\end{align*}
	by a simple double-counting argument: $(i)$ each edge from $\Vstar \setminus \Vstar_{\Delta}$ to $\Vstar_{\Delta}$ is counted once in both $\sum_{v \in \Vstar} \degDelta(v)$ and $\card{\Vstar_\Delta} \cdot \Delta$; while, $(ii)$ each edge
	with both endpoints in $\Vstar_{\Delta}$ is counted twice in both of them (and other edges are ignored by both). Thus, the two terms are equal. 
	\end{itemize}
	This concludes the proof. 
\end{proof}

To be able to apply~\Cref{fact:random-walk}-\eqref{rw-fact:balanced}, we also need to prove that the $M$-digraph is strongly connected, which is a simple observation captured in the following. 

\begin{observation}\label{obs:strongly-connected}
	For any covered matching $M \subseteq E$ with $\unmatchDelta{M} > 0$, the $M$-digraph $\GG$ is strongly connected. 
\end{observation}
\begin{proof}
	By~\Cref{obs:mc-upper-bound} and the definition of $G^*$, 
	all vertices in $\GG$ are reachable from $s$ (as $\unmatchDelta{\Mstar}=\unmatchDelta{M} > 0$, we know $s$ has an outgoing edge). For the same reason, all vertices can reach $t$. Finally, since $t$ is connected
	to $s$, we get $\GG$ is strongly connected.  
\end{proof}

Finally, we also need an upper bound on the number of edges in $M$-digraphs.

\begin{claim}\label{clm:M-digraph-edge}
	For any covered matching $M \subseteq E$, the number of edges in the $M$-digraph $\GG$ is at most 
	\[
		m(\GG) \leq 7n_{\Delta} \cdot \Delta.
	\]
\end{claim}
\begin{proof}
	We listed the out-degree of each vertex in $\GG$ in the proof of~\Cref{lem:balanced}; we use the same bounds here as well. We have, 
	\begin{align*}
		m(\GG) &= \outdeg_\GG(s) + \hspace{-8pt}\sum_{\alpha(v) \in \VVstart}\hspace{-5pt} \outdeg_{\GG}(\alpha(v)) + \hspace{-8pt}\sum_{\beta(u,v) \in \VVmatch}\hspace{-8pt} \outdeg_{\GG}(\beta(u,v)) + \hspace{-8pt}\sum_{\gamma(v) \in \VVend}\hspace{-5pt} \outdeg_{\GG}(\gamma(v)) + \outdeg_{\GG}(t) \\
		&= \paren{2 \card{\Vstar_{\Delta} \setminus \Vstar(M)} + \card{\Vstar(M)}} \cdot \Delta + \hspace{-12pt}\sum_{v \in \Vstar \setminus \Vstar(M)}\hspace{-12pt} \degDelta(v) 
		+ \paren{\card{\Vstar_\Delta} + \card{\Vstar(M) \setminus \Vstar_{\Delta}}} \cdot \Delta - \hspace{-8pt}\sum_{v \in \Vstar(M)} \hspace{-8pt}\degDelta(v) \\
		&= \paren{2\card{\Vstar_{\Delta} \setminus \Vstar(M)}+ \card{\Vstar(M)} + \card{\Vstar_\Delta} + \card{\Vstar(M) \setminus \Vstar_{\Delta}}} \cdot \Delta + \sum_{v \in \Vstar} \degDelta(v)  
		 - 2\hspace{-8pt}\sum_{v \in \Vstar(M)}\hspace{-8pt}\degDelta(v) 
		\tag{by re-organizing the terms and splitting the sum}\\
		&=  \paren{2\card{\Vstar_{\Delta} \setminus \Vstar(M)}+ \card{\Vstar(M)} + 2 \card{\Vstar_\Delta} + \card{\Vstar(M) \setminus \Vstar_{\Delta}}} \cdot \Delta - 2\hspace{-8pt}\sum_{v \in \Vstar(M)}\hspace{-8pt}\degDelta(v) 
		\tag{using the equation $\sum_{v \in \Vstar} \degDelta(v) = \card{\Vstar_{\Delta}} \cdot \Delta$ established in the proof of~\Cref{lem:balanced}} \\
		&\leq 7n_{\Delta} \cdot \Delta,
	\end{align*}
	since $\card{\Vstar_{\Delta}} \leq n_{\Delta}$ and $\card{\Vstar(M)} \leq 2n_{\Delta}$ as $M$ is covered by $V_{\Delta}$. 
\end{proof}

We are now ready to prove~\Cref{lem:rw-properties}. 

\begin{proof}[Proof of~\Cref{lem:rw-properties}]
	Let $\pi$ denote the {stationary distribution} of the standard random walk on the $M$-digraph $\GG$ (as in~\Cref{obs:mc-upper-bound}). 
	By~\Cref{lem:balanced}, the $M$-digraph is balanced and by~\Cref{obs:strongly-connected}, it is strongly connected. Hence, by~\Cref{fact:random-walk}-\eqref{rw-fact:balanced}, 
	\begin{align}
		\forall \nu \in \VV \quad \pi_{\nu} = \frac{\outdeg_{\GG}(\nu)}{m(\GG)}. \label{eq:stationary}
	\end{align}
	 We now prove each part separately. 
	
	\paragraph{Proof of Part~\eqref{part:rw-properties-1}.} Let $h_{s,t}$ denote the expected hitting time of $s$ to $t$ in $\GG$ and $h_{s,s}$ denote the expected return time of $s$. We have, 
	\[
		\Exp\card{W} ~\Eq{(1)}~ (h_{s,t}-2) ~\Leq{(2)}~ h_{s,s} ~\Eq{(3)}~ \frac{1}{\pi_s} ~\Eq{(4)}~ \frac{m(\GG)}{\unmatchDelta{\Mstar} \cdot \Delta} ~\Leq{(5)}~ \frac{7n_{\Delta}}{\unmatchDelta{M}}. 
	\]
	where: (1) is by~\Cref{obs:mc-upper-bound}, (2) is because the only incoming edges of $s$ are from $t$, (3) is by \Cref{fact:random-walk}-\eqref{rw-fact:ftmc}, (4) is by~\Cref{eq:stationary} and out-degree of $s$ in $\GG$, 
	and (5) is by~\Cref{clm:M-digraph-edge} and because we have that $\unmatchDelta{\Mstar}=\unmatchDelta{M}$. 
	
	\paragraph{Proof of Part~\eqref{part:rw-properties-2}.} This part is straightforward given that the first vertex of the walk is chosen uniformly at random from the unmatched vertices in $V_{\Delta}$.
	
	\paragraph{Proof of Part~\eqref{part:rw-properties-3}.} Fix a matched vertex $v \in V_{\Delta}$. If $v$ is not in $\Gstar$, it cannot ever appear in any Matching Random Walk (by the definition of $\Gstar$) and we are done. 
	Now suppose $v$ is in $\Gstar$. By~\Cref{obs:mc-upper-bound}, for $v$ to appear in $odd(W)$, we need $\beta(M(v),v)$ to belong to the 
	corresponding random walk $\WW$ on $\GG$ from $s$ to $t$ (for $\beta(v,M(v))$, $v$ appears in $even(W)$ instead). Thus,
	\begin{align*}
		\expect{\text{$\#$ of times $v$ appears in $odd(W)$}} &= \expect{\text{$\#$ of times $\beta(M(v),v)$ appears in $\WW$}} \\
		&\Leq{($\star$)} \frac{\pi_{\beta(M(v),v)}}{\pi_s}  \tag{by~\Cref{fact:random-walk}-\eqref{rw-fact:hit-before-return}; see below} \\
		&= \frac{\outdeg_{\GG}(\beta(M(v),v)}{\outdeg_{\GG}(s)} \tag{by~\Cref{eq:stationary}} \\
		&= \frac{1}{\unmatchDelta{M}} \tag{by the out-degrees of these vertices}. 
	\end{align*}
	Finally, the inequality $(\star)$ holds because a random walk in $\GG$ starting from $s$ and returning to it for the first time certainly visits $t$ also, because $t$ is the only vertex having an edge to $s$. 
	Hence, we can upper bound the number of occurrences of $v$ in $\WW$ by the number of occurrences of $v$ in a random walk starting at $s$ and returning to $s$ itself; this allows us to use~\Cref{fact:random-walk}-\eqref{rw-fact:hit-before-return}. 
\end{proof}

\subsection{Extracting Alternating Paths from the Matching Random Walk} 

The next step is to design a procedure for extracting an alternating path from a Matching Random Walk. 
We first state the main properties of this subroutine and then describe its implementation. 

\begin{lemma}\label{lem:extract}
	There is a randomized subroutine $\extract$ that given a Matching Random Walk $W$ with respect to a covered matching $M$, outputs 
	an alternating path $P$ for $M$ such that the matching $M' := M \triangle P$ obtained by applying $P$ to $M$ satisfies the following properties (all probabilities are over 
	the randomness of $W$ and the inner randomness of $\extract$): 
	\begin{enumerate}
		\item\label{part:extract-size} For any matching $M$, $\unmatchDelta{M'} \leq \unmatchDelta{M}$ deterministically. 
		
		\item\label{part:extract-unmatch} For any matching $M$ and any vertex $v \in V_{\Delta}$ \underline{matched} by $M$,
		\[
			\Pr_{}\paren{\text{$v$ is unmatched by $M'$}} \leq \frac{1}{\unmatchDelta{M}} \cdot \paren{\frac{1}{\Delta} + \frac{1}{\Delta^2}}. 
		\]
		\item\label{part:extract-match} For any matching $M$ and any vertex $v \in V_{\Delta}$ \underline{unmatched} by $M$, 
		\[
			\Pr_{}\paren{\text{$v$ is matched by $M'$}} \geq \frac{1}{\unmatchDelta{M}}. 
		\]
	\end{enumerate}
\end{lemma}

The general idea of the algorithm in~\Cref{lem:extract} is as follows. We will follow the walk $W$ one vertex at a time as long as it remains an alternating path, and label the 
vertices as `odd' or `even' depending on their position in the walk. Once the walk creates a cycle and is no longer a path, we will check if it is an even cycle or an odd one, based on their labels, and 
run a proper subroutine (to be described later) to ``fix'' this cycle: this means either removing it in case of an even cycle or ``rotating'' it in case of odd cycle, which roughly speaking means finding a way to 
rotate some portion of the path to ensure it remains an alternating path. For technical reasons, the algorithm also has a simple and \emph{``benign''} stopping criteria which corresponds to terminating the walk 
with probability $1/\Delta^2$ at each step and returning the currently computed path already. The formal algorithm is as follows (subroutines $\fixeven$ and $\fixodd$ will be described later). 

\begin{Algorithm}
	 $\extract$ subroutine for a Matching Random Walk $W=(v_1,v_2,\ldots)$. 
	\vspace{-5pt}
	\begin{itemize}
		\item Let the list $P = (v_1)$ initially and \textbf{label} the vertex $v_1$ in $P$ as `odd'. 
		\item For $i=1$ to $\floor{(\card{W}-1)/2}$ (in the following, $v_{2i-1}$ will be the ``head'' of the path):
		\begin{enumerate}
			\item\label{line:random-terminate} if $i \neq 1$, with probability $1/\Delta^2$, \textbf{terminate} and return $P$. 
			\item If $v_{2i}$ is not labeled, label it `even' and do the following: 
			\begin{enumerate}
			\item\label{line:augment} Append $v_{2i}$ to end of $P$. If $v_{2i}$ is the last vertex of $W$, \textbf{terminate} and return $P$. 
			\item Label $v_{2i+1}$ as `odd' and append $v_{2i+1}$ to the end of $P$. 
			\end{enumerate}
			\item If $v_{2i}$ is already labeled `even', run $\fixeven(P,v_{2i})$, and if `odd', run $\fixodd(P,v_{2i})$. 
		\end{enumerate}
	\end{itemize}
\end{Algorithm}
\begin{remark}\label{rem:random-termination}
Line~\eqref{line:random-terminate} of~$\extract$ is added to implement a ``random truncation'' approach which is needed to bound the runtime of the algorithm (and has minimal effect on its other guarantees, unlike a 
deterministic truncation that may violate the ``fairness'' property of its output matching). 
\end{remark}

Before getting to describe the subroutines $\fixeven$ and $\fixodd$, we state the main invariant maintained by $\extract$ in its iterations. 

\begin{invariant}\label{inv:extract}
	In $\extract$, at the beginning of each iteration $i$ of the for-loop:
	\begin{enumerate}
	\item the list $P = (u_1,u_2,\ldots,u_{2k+1})$ in the algorithm is an alternating path for $M$ with an odd number of vertices that starts from the vertex $v_1$ of $W$ and ends in the vertex $v_{2i-1}$ of $W$;
	\item the vertices $u_1,u_3,u_5,\ldots,$ in $P$ are labeled `odd' and vertices $u_2,u_4,u_6,\ldots$ are labeled `even' and no other vertex in the graph is labeled. 
	\end{enumerate}
\end{invariant}

We emphasize that $P$ is only a subset of the walk $W$ and moreover both $\fixeven$ and $\fixodd$ may delete some parts from $P$ (or even relabel its vertices in the latter case), hence, the need
for maintaining the invariant explicitly (it is easy to see that if $\fixeven$ and $\fixodd$ are never called, the algorithm trivially maintains the invariant by design).  

We now describe the two main subroutines. See~\Cref{fig:even-odd} for an illustration of each case. 

\paragraph{FixEven subroutine.} This is the easy case when the next vertex of $W$ creates an \emph{even} cycle in $P$. Handling this case is by removing the cycle from $P$ first and then continuing as before. 

\begin{Algorithm}
 $\fixeven$ subroutine for a labeled path $P$ and vertex $v_{2i}$ (see~\Cref{fig:even-odd}-(a),-(b)).  
	\begin{enumerate}
		\item Let $u_{2j} = v_{2i}$ be the vertex in $P=(u_1,\ldots,u_{2k+1})$ which was labeled `even' already. 
		\item Remove $u_{2j+2},u_{2j+3},\ldots,u_{2k+1}$ from $P$ and remove all their labels.
	\end{enumerate}
\end{Algorithm}

\paragraph{FixOdd subroutine.} Unlike the previous one, here, the next vertex of $W$ creates an \emph{odd} cycle and we cannot simply remove the cycle from $P$ as the new path no longer will be 
an alternating path. Instead, we have to find a way to ``rotate'' the current path $P$ so that we can continue extending it via the vertices of the walk $W$; this step is inspired by 
what the classical algorithms for non-bipartite matching do. 
Our way of finding the way to rotate the path
involves finding edges that are \emph{not} part of the walk, hence, the need for the second part of the access model in~\Cref{thm:random-walk-matching}. 

\begin{Algorithm}
 $\fixodd$ subroutine for a labeled path $P$ and vertex $v_{2i}$ (see~\Cref{fig:even-odd}-(c),-(d)).
	\begin{enumerate}
		\item Let $u_{2j+1} = v_{2i}$ be the vertex in $P=(u_1,u_2,\ldots,u_{2k+1})$ which was labeled `odd' already. 
		\item\label{line:unspecified} If there is any vertex $u_{2\ell+1}$ labeled `odd' in $N(v_{2i-1})$ for $1\leq \ell \leq j-1$, do the following: 	
		\begin{enumerate}
			\item Update $P$ to become $(u_1,\ldots,u_{2\ell+1}, u_{2k+1}, u_{2k}, \ldots, u_{2j+1}, u_{2j})$; this means \emph{removing} $u_{2\ell+2},\ldots,u_{2j-1}$ from $P$ (and their labels), \emph{reversing} the direction of $u_{2j},\ldots,u_{2k+1}$, 
			and \emph{relabeling} $u_{2k+1}$ as `even', $u_{2k}$ as `odd', and so on until the end of the path.
			\item Stop running $\fixodd$ and return to $\extract$. 
		\end{enumerate}
		\item If no such vertex $u_{2\ell+1}$ was found, \textbf{terminate} $\extract$ altogether and return $P$. 
	\end{enumerate}
\end{Algorithm}
\begin{remark}\label{rem:fixodd-unspecified}
We emphasize right away that unlike our previous subroutines,~\fixodd is not yet fully specified: we have not specified in Line~\eqref{line:unspecified} how to find the vertex $u_{2\ell+1}$ and how to choose one if there are more than one choices. 
In the following, we show that \underline{any} choice of $u_{2\ell+1}$ results in the advertised guarantees of~\Cref{lem:extract}. We then use this flexibility in~\Cref{sec:implementation} when going over the exact implementation details of our algorithms 
to implement $\fixodd$ efficiently. 
\end{remark}

\begin{figure}[p!]
	\centering
	\subcaptionbox{$\fixeven$ is called whenever $v_{2i}$ is already marked `even' and is, say, vertex $u_{2j}$ in $P$.}%
	[1\linewidth]{
		\includegraphics[scale=0.45]{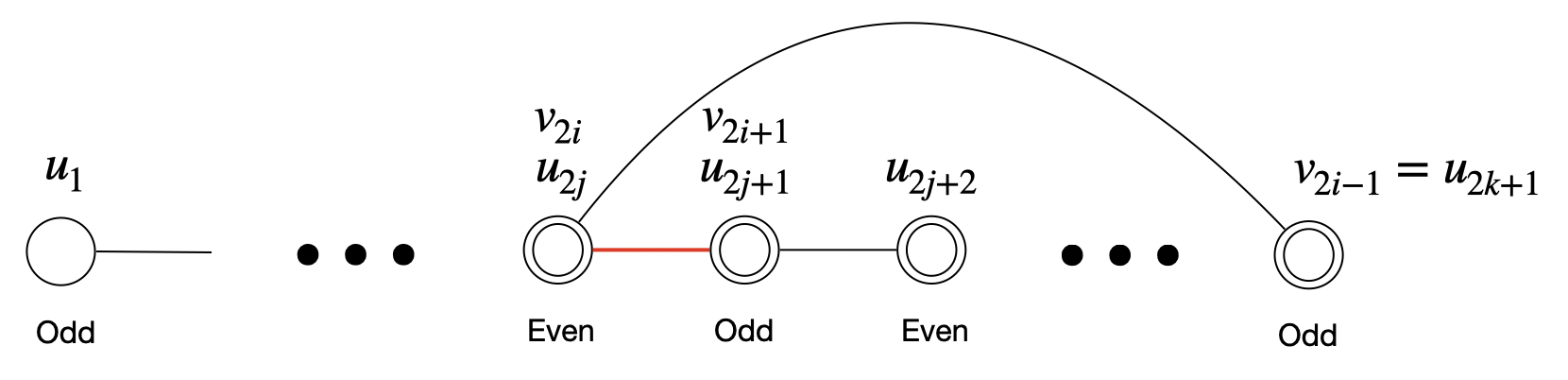}
	}

	\subcaptionbox{$\fixeven$ removes the \emph{even} cycle $u_{2j}$ to $u_{2k+1}$ from $P$ and $\extract$ continues building $P$ from $v_{2i+2}$.}%
	[1\linewidth]{
	\includegraphics[scale=0.45]{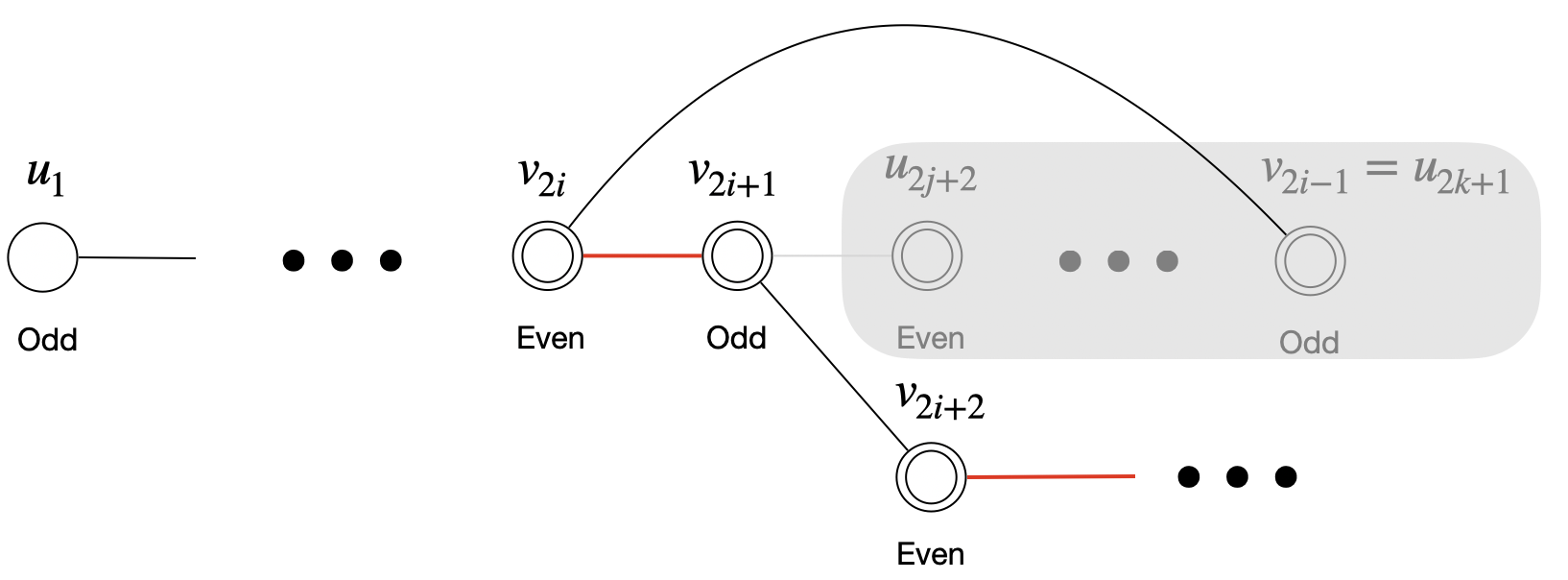}
	}
	
	\vspace{10pt}
	
	\subcaptionbox{$\fixodd$ is called whenever $v_{2i}$ is already marked `odd' and is, say, vertex $u_{2j+1}$ in $P$. The dashed (blue) edge may not be part of the walk and is instead found by $\fixodd$ in its attempt to rotate the found odd cycle.}%
	[1\linewidth]{
		\includegraphics[scale=0.45]{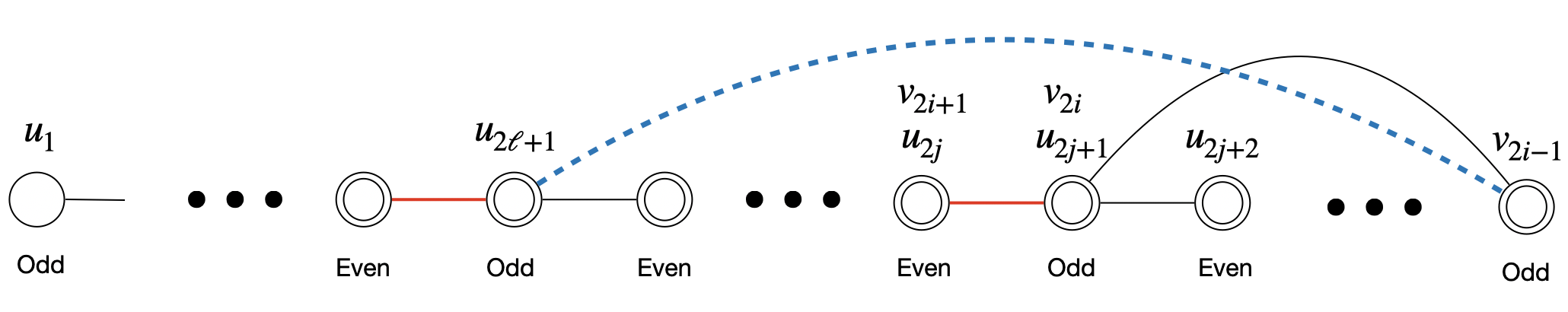}
	} 	
	
	\vspace{10pt}
	
	\subcaptionbox{$\fixodd$ attempts to find another `odd'-to-`odd' edge for $v_i=u_{2k+1}$ to some $u_{2\ell+1}$ for $\ell < j$; if it fails, it terminates the entire $\extract$. Otherwise, it removes the sub-path $u_{2\ell+2},\ldots,u_{2j-1}$ from $P$ and reverse the direction of the path afterwards by taking the edge $(u_{2\ell+1},u_{2k+1})$ first and then going backwards toward $u_{2j} = v_{2i+1}$.
	Finally, it lets $\extract$ continue building the path $P$ from $v_{2i+1}$.}%
	[1\linewidth]{
	\includegraphics[scale=0.45]{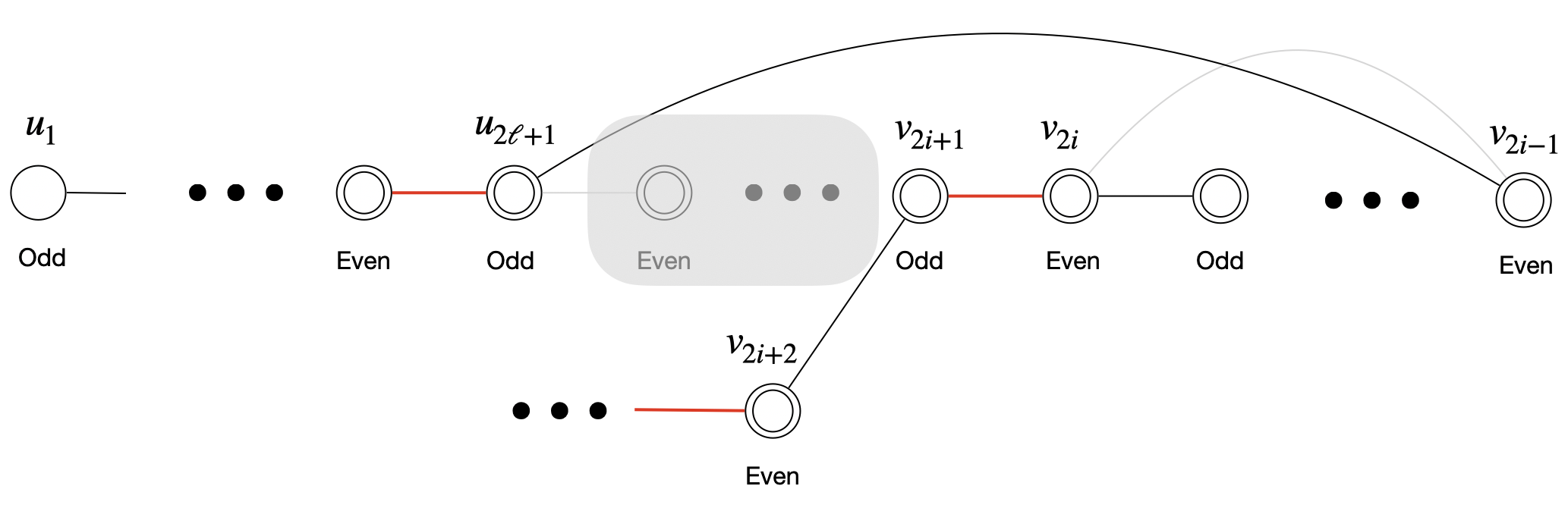}
	}

	\caption{An illustration of $\extract$ when $v_{i+1}$ is labeled previously and the algorithm calls $\fixeven$ or $\fixodd$. Double-circles mark matched vertices and single circle mark unmatched ones. Thick (red) edges are matching edges.}
	\label{fig:even-odd}
\end{figure}

We now prove that~\Cref{inv:extract} holds throughout the run of $\extract$ (see also~\Cref{fig:even-odd}).  

\begin{proof}[Proof of~\Cref{inv:extract}]
	The proof is by induction. The invariant holds for iteration $1$ of the for-loop by considering the singleton vertex $v_1$ of $W$ as an alternating path. 
	Suppose the invariant holds at the beginning of iteration $i$ and we show it holds for iteration $i+1$ as well. In the following, we ignore the cases when the iteration terminates $\extract$ (i.e., in Line~\eqref{line:augment}
	or in the call to $\fixodd$ if no $u_{2\ell+1}$ is found) since there will be no next iteration to maintain the invariant for.

	\paragraph{Case 1: When $v_{2i}$ is not labeled.} We know $(v_{2i-1},v_{2i})$ is in $E \setminus M$ in this case because if it belonged to $M$, 
	then $v_{2i} = M(v_{2i-1}) = v_{2i-2}$ (by induction) and thus it would have already been labeled `even' (and so will be handled by the next case).  
	Also, $M(v_{2i})$ cannot be in $P$ because $P$ is an alternating path and if a vertex belongs to $P$ its matched pair should also be in $P$. Thus, $(v_{2i},v_{2i+1})$ is an edge in $M$ that is not in  $P$. 
	Hence, appending $(v_{2i},v_{2i+1})$ to $P$ correctly extends the alternating path $P$ to end in $v_{2i+1}$ as desired for iteration $i+1$. The labels are extended trivially in this case
	and thus the invariant continues to hold. 
	 
	\paragraph{Case 2-even: When $v_{2i}$ is labeled `even'.} In this case (handled by $\fixeven$), $u_{2j} = v_{2i}$ and $u_{2j+1} = v_{2i+1}$ since $P$ is an alternating path and $u_{2j+1} = M(u_{2j}) = M(v_{2i}) = v_{2i+1}$ (the latter by the design of $W$). 
	By removing $u_{2j+2},\ldots,u_{2k+1}$ from $P$ in $\fixeven$, we obtain an alternating path ending with $u_{2j} (=v_{2i})$ and $u_{2j+1} (v=v_{2i+1})$, which maintains 
	the invariant for the iteration $i+1$. 
	
	\paragraph{Case 2-odd: When $v_{2i}$ is labeled `odd'.} In this case (handled by $\fixodd$), $u_{2j+1} = v_{2i}$ and $u_{2j} = v_{2i+1}$ since $P$ is an alternating path and $u_{2j} = M(u_{2j+1}) = M(v_{2i}) = v_{2i+1}$ (the latter by the design of $W$). 
	We claim that the new path $P = (u_1,\ldots,u_{2\ell+1}, u_{2k+1}, u_{2k}, \ldots, u_{2j+1}, u_{2j})$ is an alternating path for $M$. This is because $u_1,\ldots,u_{2\ell+1}$ is an alternating path already, 
	$(u_{2\ell+1},u_{2k+1})$ is an edge in $E \setminus M$ (it cannot be in $M$ because $M(u_{2\ell+1}) = u_{2\ell} \neq u_{2k+1}$), $(u_{2k+1},u_{2k})$ is an edge in $M$, and continuing like this we have alternating
	edges in and out of $M$. Finally, $u_{2j} = v_{2i+1}$ and thus the path ends in the desired vertex for iteration $i+1$ ($P$ also has an odd number of vertices clearly). Finally, we 
	are also explicitly relabeling all the vertices in $P$ again. 
	
	In conclusion, in every case,~\Cref{inv:extract} is maintained, completing the proof. 
\end{proof}

We can now start analyzing $\extract$. The next claim lists some of its key properties. 

\begin{claim}\label{clm:extract}
	Given any Matching Random Walk $W$ with respect to a covered matching $M$, $\extract(W)$ returns an alternating path $P$ such that:  
	\begin{enumerate}
		\item the first vertex $v_1$ of $W$ is always matched by $M \triangle P$; 
		\item the matching $M \triangle P$ has the same number or larger number of matches vertices in $V_{\Delta}$.  
	\end{enumerate}
\end{claim}
\begin{proof}
	By~\Cref{inv:extract}, the path $P$ computed by the algorithm at the start of each iteration $i$ is an alternating path ending in $v_{2i-1}$. We use this in the following to 
	prove each property separately:  
	\begin{enumerate}
		\item The alternating path $P$ starts with $v_1$ in the first iteration. The termination in Line~\eqref{line:random-terminate} is not run 
		when $i=1$ so $P$ always starts with $v_1$ followed by at least one more edge. When some sub-path of $P$ is removed in $\fixeven$, we remove $u_{2j},\ldots$ for some $j$ and thus the first 
		vertex of the path is not touched. Similarly, when some sub-path of $P$ is removed in $\fixodd$, we remove some $u_{2\ell+2},\ldots,u_{2j-1}$ for some $1 \leq \ell < j$ and thus again the first vertex of the path is not touched. 
		Hence, $v_1$ always remain part of $P$ and thus is matched by $M \triangle P$. 
		\item  Either $P$ is an augmenting path which means $M \triangle P$ matches all vertices of $M$ plus two new ones ($v_1$ and $v_{2i-1}$), 
		or it is an alternating path and so $v_1$ will be matched and $v_{2i-1}$ unmatched. Since $v_1$ is in $V_{\Delta}$ and $v_{2i-1}$ may or may not be, this does not decrease the number of matched vertices in $V_{\Delta}$. \qed
	\end{enumerate}
\end{proof}

We now prove each part of~\Cref{lem:extract} separately (we note that the main part is~\eqref{part:extract-unmatch}). 

\begin{proof}[\emph{\textbf{Proof of~\Cref{lem:extract}-Part~\eqref{part:extract-size}}}]
Follows immediately from~\Cref{clm:extract} because in every possible case of termination in $\extract$, we have $\unmatchDelta{M'} \geq \unmatchDelta{M}$. \qed{\scriptsize ~Part~\eqref{part:extract-size}}

\end{proof}

\begin{proof}[\emph{\textbf{Proof of~\Cref{lem:extract}-Part~\eqref{part:extract-unmatch}}}]

	We next prove that for any vertex $v \in V_{\Delta}$ {matched} by $M$,
	\[
			\Pr_{}\paren{\text{$v$ is unmatched by $M'=M\triangle P$}} \leq \frac{1}{\unmatchDelta{M}} \cdot \paren{\frac{1}{\Delta} + \frac{1}{\Delta^2}}. 
	\]
	The only time the walk $W$ may result in a vertex $v$ in $V_{\Delta}$ to no longer be matched is either when (a) the random termination condition of Line~\eqref{line:random-terminate} happens, or (b) 
	the algorithm terminates in $\fixodd$ when processing this vertex $v$ 
	(in either cases, $v$ will be the last vertex of the alternating path and thus will become unmatched). 
The next claim bounds the probability of this event happening. 
	
		\begin{claim}\label{clm:fixodd-unmatch}
	For any matching $M$ and any vertex $v \in V_{\Delta} \cap V(M)$, 
	\[
		\Pr_{}\paren{\text{$\extract(W)$ terminates when processing $v$}} \leq \expect{\text{$\#$ of times $v$ appears in $odd(W)$}} \cdot \paren{\frac{1}{\Delta}+\frac{1}{\Delta^2}}.
	\]
\end{claim}
\begin{proof}
	Let $I = \floor{(W-1)/2}$ be the random variable for the \emph{maximum} number of iterations in $\extract(W)$. For every $i \in [I]$ we define the events: 
	\begin{itemize}
		\item $\event_{}(i)$: the event that $\extract$ terminates in Line~\eqref{line:random-terminate} or $\fixodd$ in iteration $i$; 
		\item $\event_v(i)$: the event that the vertex $v_{2i-1}$ of $W$ (considered in iteration $i$ to choose $v_{2i} \in N(v_{2i-1})$)
	is the vertex $v$ in the claim statement, i.e., $v_{2i-1} = v$. 
	\end{itemize}
	Consequently, we have, 
	\begin{align}
		\Pr_{}\paren{\text{$\extract(W)$ terminates when processing $v$}} &= \sum_{i=1}^{\infty} \Pr_{W}\Paren{\event_{}(i) \wedge \event_v(i)}, \label{eq:RHS2}
	\end{align}
	given these events are all disjoint (the algorithm can only terminate once). 
	
	Clearly, the probability of terminating due to Line~\eqref{line:random-terminate} is simply $1/\Delta^2$. We now bound the probability of terminating due to $\fixodd$. 
	
	In $\extract$, the choice of $P$ at the beginning of iteration $i$ only depends on the choice of $v_1,\ldots,v_{2i-1}$ in $W$. Thus, at this stage, $v_{2i}$ is still chosen uniformly at random 
	from neighbors of $v_{2i-1}$ conditioned on $P$. We claim that: 
	\begin{quote}
	There is only \textbf{one choice} for $v_{2i}$ in $N(v_{2i-1})$ that forces $\extract(W)$ to terminate in $\fixodd$ in iteration $i$.
	
	\emph{Proof.} This is the neighbor $u$ of $v_{2i-1}$ that is labeled `odd' and is closest to $v_1$ in $P$ among its `odd' neighbors; any other neighbor is 
	either not labeled, labeled `even', or choosing that vertex as $v_{2i}$ allows for finding an index $\ell$ in $\fixodd$ and thus not terminating (the index $\ell$ can then be the index of $u$ or possibly some other vertex closer to $v_{2i-1}$). 
	\end{quote}

	Finally, if there is a choice for $v_{2i}$ to be made (which leads the algorithm to call $\fixodd$), 
	we have $v_{2i-1} \in V_{\Delta}$ as otherwise the walk already finishes. Thus, for every $i \geq 1$, and any choice of first $2i-1$ vertices $v_1,\ldots,v_{2i-1}$ in $W$,
	\begin{align}
		\Pr_{}\paren{\event_{}(i) \mid v_1,\ldots,v_{2i-1} \in W \wedge v_{2i-1} \in V_{\Delta}} \leq \frac{1}{\Delta^2} + \frac{1}{\deg_G(v_{2i-1})} = \frac{1}{\Delta} + \frac{1}{\Delta^2}; \label{eq:stop-i}
	\end{align}
	the final bound is because $v_{2i}$ is chosen uniformly from $N(v_{2i-1})$ and $v_{2i-1}$ belongs to $V_{\Delta}$ and thus its degree is $\Delta$. Using this, we can bound the RHS of~\Cref{eq:RHS2} as follows: 
	\begin{align*}
		\sum_{i=1}^{\infty} \Pr_{}\Paren{\event_{}(i) \wedge \event_v(i)} &=  \sum_{i=1}^{\infty} \Pr_{}\Paren{\event_{}(i) \wedge \event_v(i) \wedge I \geq i} \tag{if $I < i$, then the walk has already terminated and $\event_{}(i),\event_v(i)$ cannot happen} \\
		&=  \sum_{i=1}^{\infty} \Pr_{}\Paren{\event_{}(i) \mid \event_v(i) \wedge I \geq i} \cdot \Pr\paren{\event_v(i) \wedge I \geq i} \\
		&\leq \paren{\frac{1}{\Delta}+\frac{1}{\Delta^2}} \cdot \sum_{i=1}^{\infty} \Pr\paren{\event_v(i) \wedge I \geq i} \tag{by~\Cref{eq:stop-i} as $\event_v(i)$ and $I \geq i$ are only functions of $v_1,\ldots,v_{2i-1}$ in $W$ and $I \geq i$ means $v_{2i-1} \in V_{\Delta}$} \\
		&= \paren{\frac{1}{\Delta}+\frac{1}{\Delta^2}} \cdot \sum_{i=1}^{\infty} \Pr\paren{\text{vertex $v_{2i-1}$ of $W$ is $v$}}
		\tag{the events $\event_v(i)$ and $I \geq i$ happen iff $W$ has at least $2i$ vertices and the $(2i-1)$-th vertex is $v$} \\
		&= \paren{\frac{1}{\Delta}+\frac{1}{\Delta^2}} \cdot \expect{\text{$\#$ of times $v$ appears in $odd(W)$}}. \Qed{clm:fixodd-unmatch}
	\end{align*}
	
\end{proof}

We now have, 
\begin{align*}
\Pr_{}\paren{\text{$v$ is unmatched by $M'$}} &= \Pr_{}\paren{\text{$\extract(W)$ terminates when processing $v$}} \\
		&\leq \paren{\frac{1}{\Delta}+\frac{1}{\Delta^2}} \cdot\expect{\text{$\#$ of times $v$ appears in $odd(W)$}} \tag{by~\Cref{clm:fixodd-unmatch}} \\
		&\leq \paren{\frac{1}{\Delta}+\frac{1}{\Delta^2}} \cdot \frac{1}{\unmatchDelta{M}} \tag{by~\Cref{lem:rw-properties}},   
\end{align*}
proving Part~\eqref{part:extract-unmatch} of~\Cref{lem:extract}. \qed{\scriptsize ~Part~\eqref{part:extract-unmatch}}

\end{proof}

\begin{proof}[\emph{\textbf{Proof of~\Cref{lem:extract}-Part~\eqref{part:extract-match}}}]
	Finally, we prove that for any vertex $v \in V_{\Delta}$ {unmatched} by $M$,
	\[
			\Pr_{}\paren{\text{$v$ is matched by $M'$}} \geq \frac{1}{\unmatchDelta{M}}. 
	\]
	This is  because by~\Cref{clm:extract}, the first vertex $v_1$ of $W$ is matched by $M'$ and by~\Cref{lem:rw-properties}, 
	any of the $\unmatchDelta{M}$ unmatched vertices in $V_{\Delta}$ have the same probability of being $v_1$. \qed{\scriptsize ~Part~\eqref{part:extract-match}}
	
\end{proof}


\subsection{The \rwm Algorithm}

We are now ready to present the $\rwm$ algorithm in~\Cref{thm:random-walk-matching}. The algorithm starts with an empty matching $M$ and repeat the following process: compute a Matching Random Walk $W$ with
respect to $M$; extract an alternating path $P = \extract{(W)}$ from this walk; apply this alternating path $P$ to $M$ to update it and continue like this. 

Furthermore, to determine when the algorithm should terminate, 
we introduce a \emph{budget} for it: the algorithm starts with $b=0$ and each time it computes a walk $W$, it increases $b$ by $1/\unmatchDelta{M}$ and keep running the algorithm until $b$ become at least $2\ln{\Delta}$. The budget roughly captures the fact that when $\unmatchDelta{M}$ is large, the walk $W$ is shorter (by~\Cref{lem:rw-properties}) and thus can be computed quicker than when $\unmatchDelta{M}$ has become
small; thus, the algorithm is allowed to spend more iterations for shorter walks compared to the longer ones. The formal algorithm is as follows. 

\begin{Algorithm}\label{alg:rwm}
	 $\rwm$ algorithm for a given $(V,E,\Delta,V_{\Delta})$. 
	 \vspace{-5pt}
	\begin{enumerate}
		\item Initialize the matching $M =\emptyset$ and the \textbf{budget} $b = 0$. 
		\item While $b < 2\ln{\Delta}$: 
		\begin{enumerate}
		\item Update $b \leftarrow b + (1/\unmatchDelta{M})$.
		\item Let $W$ be a Matching Random Walk with respect to $M$ and $P = \extract{(W)}$ be an alternating path. Apply $P$ to $M$ and update $M \leftarrow M \triangle P$. 
		\end{enumerate}
	\end{enumerate}
\end{Algorithm}

The following lemma establishes the ``fairness'' of the matching $M$ output by $\rwm$. 

\begin{lemma}\label{lem:prob-unmatched}
	For any vertex $v \in V_{\Delta}$ and matching $M$ returned by $\rwm$,
	\[
		\Pr\paren{\text{$v$ is unmatched by $M$}} \leq \frac{1}{\Delta} + \frac{2}{\Delta^2}. 
	\]
\end{lemma}
\begin{proof}
	For any iteration $i \geq 1$ of the while-loop in $\rwm$, we use $M_i$ to denote the matching $M$ at the \emph{start} of this iteration and $P_i$ be the alternating path computed in this iteration. 
	Thus, $M_{i+1} = M_i \triangle P_i$. We let $k$ denote the index of the last iteration of the while-loop, thus the matching $M$ returned by $\rwm$ is $M_{k+1}$.  Note that, while a random variable, $k$ is always finite as each iteration 
	of the while-loop increases the budget by at least $1/n_{\Delta}$. For any vertex $v \in V_{\Delta}$, 
	\begin{align}
		&\Pr\paren{\text{$v$ is unmatched in $M$}} \notag \\
		&\hspace{10pt}= \Pr\paren{\text{$v$ was never matched by any $M_i$ for $i \in [k+1]$}} \notag \\
		&\hspace{20pt} + \sum_{i=1}^{k}\Pr\paren{\text{$v$ got unmatched by $M_i \triangle P_i$ and never got matched by any $M_j$ for $j > i$}}; \label{eq:v-unmatched}
	\end{align}
	this is because these events are all disjoint (for $v$ to get unmatched by $M_i \triangle P_i$, it needs to be matched in $M_i$), and together they cover all possibilities of $v$ being unmatched in $M$ at the end (we emphasize that
	in the above terms, $k$ itself is also a random variable). 
	
	The first term of~\Cref{eq:v-unmatched} can be bounded as follows: 
	\begin{align}
		&\Pr\paren{\text{$v$ was never matched by any $M_i$ for $i \in [k+1]$}} \notag \\
		&\hspace{20pt}= \prod_{i=1}^{k} \Pr\paren{\text{$v$ not matched by $M_i \triangle P_i \mid v$ not matched by $M_i$}}  \tag{by chain rule} \\
		&\hspace{20pt}\leq  \prod_{i=1}^{k} \label{eq:v-unmatched-1} \paren{1-\frac{1}{\unmatchDelta{M_i}}}
		 \leq \exp\paren{-\sum_{i=1}^{k} \frac{1}{\unmatchDelta{M_i}}}  \tag{by~\Cref{lem:extract}-Part~\eqref{part:extract-match}}  \\
		&\hspace{20pt} = \exp\paren{-b} \leq \frac{1}{\Delta^2}. \label{eq:v-unmatched-1}
	\end{align}
	where $b \geq 2\ln{\Delta}$ is the final budget of the algorithm upon termination, by definition. 
	
	We now bound the second term of~\Cref{eq:v-unmatched}. 
	\begin{align}
		&\sum_{i=1}^{k}\Pr\paren{\text{$v$ got unmatched by $M_i \triangle P_i$ and never got matched by any $M_j$ for $j > i$}} \notag \\
		&\hspace{20pt}\leq \sum_{i=1}^{k} \frac{1}{\unmatchDelta{M_i}}  \cdot \paren{\frac{1}{\Delta}+\frac{1}{\Delta^2}} \cdot \prod_{j=i+1}^{k} \paren{1-\frac{1}{\unmatchDelta{M_j}}} \tag{by~\Cref{lem:extract}-Part~\eqref{part:extract-unmatch}
		for $i$, and the same calculation as in~\Cref{eq:v-unmatched-1} for iterations $>i$} \\
		&\hspace{20pt} \leq \paren{\frac{1}{\Delta}+\frac{1}{\Delta^2}}  \notag
	\end{align}
	where the final inequality is by~\Cref{prop:weird-inequality} (by setting $x_i := 1/\unmatchDelta{M_i} \in (0,1]$). 

	Plugging in this and~\Cref{eq:v-unmatched-1}  in the RHS of~\Cref{eq:v-unmatched} concludes the proof. 
\end{proof}

Finally, we also have the following lemma that will bound the total expected length of all the walks considered by $\rwm$ (this will be used later to bound the runtime of the algorithm). 

\begin{lemma}\label{lem:rwm-runtime}
	The expected length of all Matching Random Walks in $\rwm$ is $O(n_{\Delta} \ln{({\Delta})})$. 
\end{lemma}
\begin{proof}
	We use the same notation as in the proof of~\Cref{lem:prob-unmatched}. In particular, for any iteration $i \geq 1$, $M_i$ and $W_i$ denote, respectively, the matching and the Matching Random Walk with respect to this matching in the iteration $i$
	of the while-loop. And, $k$ is the number of iterations of the while-loop. 
	Thus, by~\Cref{lem:rw-properties}, we have, 
		\begin{align*}
		\sum_{i=1}^{k} \Exp\card{W_i} \leq \sum_{i=1}^{k} \frac{7n_{\Delta}}{\unmatchDelta{M_i}} 
		= 7n_{\Delta} \cdot \underbrace{\sum_{i=1}^{k} \frac{1}{\unmatchDelta{M_i}}}_{=\text{budget $b$}} 
		= O(n_{\Delta} \cdot \ln{(\Delta)}), 
	\end{align*}
	concluding the proof. 
\end{proof}

\subsection{Implementation Details and Concluding the Proof of~\Cref{thm:random-walk-matching}}\label{sec:implementation}

Given~\Cref{lem:prob-unmatched}, it only remains to bound the running time of~$\rwm$ to conclude the proof of~\Cref{thm:random-walk-matching}. 
To do this, we need to further specify the implementation details of each subroutine also, which we do in the following. 

\subsubsection*{Matching Random Walk} 

Given any matching $M$, a Matching Random Walk $W$ can be easily found in $O(\card{W})$ expected time given the access model of~\Cref{thm:random-walk-matching} as described below. 

Each vertex $v_{2i-1} \in odd(W)$ requires sampling $v_{2i} \in N(v_{2i-1})$ which is satisfied in the access model for each $v_{2i-1} \in V_{\Delta}$ in $O(1)$ time; but, if $v_{2i-1}$ is not in $V_{\Delta}$, 
then, by definition, $W$ terminates and thus we do not need to sample $v_{2i}$ at all. 

Each vertex $v_{2i} \in even(W)$ only requires setting $v_{2i+1} = M(v_{2i})$; thus, by storing $M$ in an array, we can implement this step in $O(1)$ time also. 

\subsubsection*{Alternating Path Subroutines} 

We now show how to implement $\extract$ and its subroutine. This step is the main part of the implementation and is non-trivial because we need to maintain the labels `odd' and `even' for the path $P$ 
\emph{implicitly}. This is because the subroutine $\fixodd$ changes these labels for (possibly long) sub-paths of $P$, but we will not have enough time to actually make these changes. 
We describe our fix in the following.

\paragraph{Data structures.} We store the path $P$ via a \emph{Treap}~\cite{AragonS89} (a.k.a. a randomized search tree). For our purpose, we think of the Treap as implementing an \emph{array} $T$ with these extra operations: 
\begin{itemize}
	\item \emph{Insert(T,x,i)}: insert $x$ in a given position $i \geq 1$ of the array $T$; 
	\item \emph{Search(T,x)}: return the index of $x$ in the array $T$ or return it does not exist; 
	\item \emph{Delete(T,i,j)}: delete all elements between indices $i \leq j$ from $T$; 
	\item \emph{Reverse(T,i,j)}: reverse the order of the elements in the sub-array $T[i:j]$; 
	\item \emph{Append(T,S)}: insert all elements in a given set $S$ to the end of the array $T$; 
	\item \emph{Pred(T,x,t)}: given an element $x$ in $T$, respectively, return the previous  $t$ elements of $T$.\footnote{A \emph{Succ(T,x,t)} operation that returns the next $t$ elements of $T$ is also possible, but 
	we do not need this operation.} 
\end{itemize}
Among these, \emph{Insert, Search, Delete}, and \emph{Reverse} take $O(\log{\card{T}})$ expected time, \emph{Append} takes $O(\card{S} + \log{\card{T}})$ expected time,
and \emph{Pred} takes $O(\log{\card{T}} + t)$ expected time. We note that the \emph{Pred} operation works as an \emph{iterator}, meaning that, after the $O(\log{\card{T}})$ expected time of preprocessing, 
we can read each of the next elements in $O(1)$ element in constant \emph{amortized} time (e.g., we can terminate the algorithm after examining $t'$ elements and only pay $O(t')$ time instead of $O(t)$). 

We note that the idea of using a Treap for implementing an array without explicitly storing its indices is quite standard: we will simply ignore the ``search-values'' of the Treap entirely and instead consider the structure of the resulting tree
as an implicit way of determining the indices (the in-order traversal of the tree gives us the indices of the elements in the array); the merge and split operations of the Treap also allows us to implement sub-array queries (\emph{Delete} and \emph{Reverse})
efficiently. This is sometimes called an \emph{Implicit Treap} (see, e.g.~\cite{ImplicitTreap} for a detailed implementation).

In addition to $T$, we also store the elements in the path $P$ in a hash table $H$ that allows insertion and deletion in $O(1)$ expected time. 
Finally, we will also have an array $S$ which is initially empty. 

\paragraph{Processing $v_{2i-1}$ in the main body of $\extract$ (i.e., \emph{not} in $\fixeven$ and $\fixodd$).} We check if the neighbor $v_{2i}$ of $v_{2i-1}$ chosen here belongs to $P$ by using the hash table $H$. If it does, 
then $v_{2i-1}$ will be processed by $\fixeven$ or $\fixodd$ and so is not considered here. Otherwise, we insert $v_{2i}$ (and possibly $v_{2i+1}$ if it exists) to the end of the array $S$ and hash table $H$ (but do not insert them to $T$ directly yet). 

\paragraph{Checking the label of $v_{2i-1}$ to call one of $\fixeven$ or $\fixodd$.} If the neighbor $v_{2i}$ of $v_{2i-1}$ chosen in this step belongs to $H$, then we know that it has already appeared in the path $P$. Hence, it needs to be handled differently. 
We first run \emph{Append(T,S)} to add vertices collected in $S$ to $T$ so that $T$ represents the entire path; set $S = \emptyset$ for the next steps. By~\Cref{inv:extract}, the indices of vertices in $T$ determine their labels 
so we run \emph{Search$(T,v_{2i})$} to find the label of $v_{2i}$ and call $\fixeven$ or $\fixodd$ accordingly. 

\paragraph{Processing $v_{2i-1}$ in $\fixeven$.} We need to delete $u_{2j+2},u_{2j+3},\ldots,u_{2k+1}$ from $P = (u_1,\ldots,u_{2k+1})$ 
where $u_{2j} = v_{2i}$ is the chosen neighbor of $v_{2i-1}$ which is already labeled `even'. We do this by calling \emph{Delete(T,$2j+2,2k+1$)}. 
We then insert $v_{2i+1}$ to the array $S$ (to be inserted later to $T$). 

\paragraph{Processing $v_{2i-1}$ in $\fixodd$.} In this case, we have $v_{2i} = u_{2j+1}$ for some $u_{2j+1}$ in the path $P$. 
We go over $\ell = j-1$ down to $1$ and find the first vertex $u_{2\ell+1}$ that also belongs to $N(v_{2i-1})$. This is done via the \emph{Pred$(T,u_{2j+1},2j+1)$} in an iterator fashion 
by getting each predecessor element one at a time and terminating once we find $u_{2\ell+1}$. Moreover, to check if a vertex in $P$ belongs to $N(v_{2i-1})$, we use the second part of the access model
in the theorem. 

In parallel, we sample $\log{\card{T}}$ neighbors of $N(v_{2i-1})$ (since $v_{2i-1} \in V_{\Delta}$, we can apply our access model). Check against 
the hash table $H$ to see if any of them belongs to $P$ and for each vertex $u'$ among those, run \emph{Search$(T,u')$} to find their index in $T$. 
If there exists some $u_{2\ell'+1}$ labeled `odd' among the samples such that $\ell' < \ell$, switch $u_{2\ell+1}$ to be $u_{2\ell'+1}$ instead (in case of more than one choice for $u_{2\ell'+1}$, pick the first one). 
Note that this is valid since $u_{2\ell'+1}$ has an `odd' label in this case and is a neighbor of $N(v_{2i-1})$. 

Finally, we remove $u_{2\ell+2},\ldots,u_{2j-1}$ from $P$ by running \emph{Delete(T,$2\ell+2,2j-1$)}, and then reverse the direction of $u_{2j},\ldots,u_{2k+1}$ in $T$ 
by running \emph{Reverse$(T,2j,2k+1)$}. 

\medskip
This concludes our implementation of $\extract$. We claim that this allows us to implement $\extract(W)$ in asymptotically $\Exp\card{W}$ time, i.e., 
\begin{align}
	\Exp_{W,R}\Bracket{\text{runtime of $\extract(W)$}} = O(1) \cdot \Exp\card{W}. \label{eq:runtime-extract}
\end{align}
The idea behind the proof is as follows. Firstly, using the Treap $T$ allows us to bound the runtime of almost all the operations with a constant (in an amortized sense), except for an additional $O(\log{\card{T}})$ time, 
whenever one of $\fixeven$ or $\fixodd$ is called\footnote{This part is already enough to argue that $\extract(W)$ can be implemented in $\Exp\card{W} \cdot O(\log{n})$ time, 
and obtain a final guarantee of $O(n_{\Delta}\cdot\log{\Delta}\cdot\log{n})$ on the runtime of $\rwm$.}. Secondly,  we show that in expectation, whenever $\fixeven$ or $\fixodd$ is called, we remove ``many'' vertices from the path $P$ (and thus Treap $T$); 
this allows us to charge the extra $O(\log{\card{T}})$ time of these steps to the prior insertion of these removed vertices to obtain an $O(1)$ amortized expected time for these steps as well. 

While the proof of this part is not particularly straightforward and requires some additional ideas, we find it more a technical challenge than an instructive argument. As such, we postpone this proof to~\Cref{app:implementation-matching}
to keep the flow of this subsection.

\subsubsection*{Implementing $\rwm$} 

Finally, the last step is to implement $\rwm$ in~\Cref{alg:rwm}. This step is completely straightforward given the implementation of prior subroutines, by following~\Cref{alg:rwm} verbatim. 
The runtime of the algorithm is then 
\begin{align}
	O(1) \cdot \sum_{i \geq 1} \paren{\Exp\card{W_i}} = O(n_{\Delta} \cdot \log{\Delta}), \label{eq:final-runtime}
\end{align}
where the first part is because a Matching Random Walk $W$ can be implemented in $O(\card{W})$ time and $\extract(W)$ runs in another $O(\card{W})$ time by~\Cref{eq:runtime-extract}; 
the equality is by~\Cref{lem:rwm-runtime} that bounds the total length of the Matching Random Walks. 

\subsubsection*{Concluding the Proof of~\Cref{thm:random-walk-matching}}

\Cref{thm:random-walk-matching} now follows from~\Cref{lem:prob-unmatched} that bounds the probability of a vertex $v \in V_{\Delta}$ remain unmatched by $1/\Delta + 2/\Delta^2$ and~\Cref{eq:final-runtime} that bounds the runtime. 

Finally,  to decrease the $2/\Delta^2$ term in the probability of a vertex remaining unmatched, to $2/\Delta^c$ for any constant $c \geq 2$, we simply do the following. Increase the budget parameter in $\rwm$ to $c \cdot \ln{\Delta}$ instead and decrease the termination probability in Line~\eqref{line:random-terminate} of $\extract$ to $1/\Delta^{c}$. 
Both these changes only increase the runtime by a $\poly(c) = O(1)$ factor (as dependence of algorithms to this parameter is poly-logarithmic). But, now the same exact proofs of~\Cref{lem:extract} and~\Cref{lem:prob-unmatched} show that the 
probability of leaving a vertex unmatched decreases to $1/\Delta + 2/\Delta^c$.

\clearpage


\newcommand{\approxcolor}{\ensuremath{\textnormal{\texttt{NearVizingColoring}}}\xspace}

\newcommand{\degpeel}{\ensuremath{\textnormal{\texttt{DegreePeeling}}}\xspace}

\newcommand{\Ecol}{\ensuremath{E_{\textnormal{\textsf{color}}}}}
\newcommand{\Erem}{\ensuremath{E_{\textnormal{\textsf{rem}}}}}
\newcommand{\Frem}{\ensuremath{F_{\textnormal{\textsf{rem}}}}}

\newcommand{\kfor}{\ensuremath{\Delta_{\textnormal{\textsf{for}}}}}
\newcommand{\kdeg}{\ensuremath{\Delta_{\textnormal{\textsf{rem}}}}}

\section{A Near-Linear Time Algorithm for ``Near''-Vizing's Coloring}\label{sec:near-vizing}

We present our main algorithm in this section, namely, an algorithm that in expected near-linear time finds an edge coloring with only $O(\log{n})$ more colors
than Vizing's theorem. 

\begin{theorem}\label{thm:fast-approx}
	There is a randomized algorithm \approxcolor that given any simple graph $G=(V,E)$ with maximum degree $\Delta$, outputs a 
	proper edge coloring of $G$ with $\Delta + O(\log{n})$ colors in $O(m\log{\Delta})$ time in expectation. 
\end{theorem}

We start with the conceptual description of the algorithm and its analysis and postpone the implementation details to~\Cref{sec:fast-approx-implement}. 
The algorithm works by partitioning the original vertices into $\Delta$ classes $V_1,\ldots,V_{\Delta}$ where $V_i$ consists of vertices with degree $i$. 
It runs $\rwm$ on $V_\Delta$  to find a matching $M_{\Delta}$; for any vertex in $V_{\Delta}$ remained unmatched, it then picks a random incident edge in $F_{\Delta}$. 
Next, both $M_{\Delta}$ and $F_{\Delta}$ are removed from the graph which reduces the maximum degree of the remaining graph to $\Delta-1$. Moreover, these max-degree vertices 
belong to the set $V_{\Delta} \cup V_{\Delta-1}$ clearly. So, the algorithm finds them in this set, and run $\rwm$ on $(\Delta-1)$-degree vertices to compute $M_{\Delta-1}$, defines $F_{\Delta-1}$ as before
and continues like this. At the end, it colors $M_{\Delta},\ldots,M_1$ with $\Delta$ different colors, and color all edges in $F := F_{\Delta} \cup \ldots \cup F_1$ 
via a simple greedy algorithm using $O(\Delta(F))$ colors. 
The formal algorithm is as follows.

\begin{Algorithm}\label{alg:fast-decomposition}
	The \approxcolor algorithm for $G=(V,E)$ with max-degree $\Delta$.  

	\begin{enumerate}
	\item For $i \in [\Delta]$, let $V_i := \set{v \in V \mid \deg(v) = i}$. Initialize $\Erem \leftarrow E$. 
	
\item For $i=1$ to $\Delta$ iterations, do the following: 
	\begin{enumerate}[leftmargin=10pt]
		\item Let $\Delta_i := \!\Delta\!-\!i\!+\!1$ and $U_i$ be the vertices in $V_{\Delta},V_{\Delta-1},\ldots,V_{\Delta_i}$ with degree $\Delta_i$ in $\Erem$. 
		\item\label{line:rwm} Run $\rwm(V,\Erem,\Delta_i,U_i)$ to obtain a matching $M_i$. 
		\item\label{line:rnei} For every $v \in U_i$ \underline{unmatched} by $M_i$, pick a random incident edge of $v$ and add it to $F_i$. 
		\item\label{line:erem} Update $\Erem \leftarrow \Erem \setminus (M_i \cup F_i)$. 
	\end{enumerate}
	\item Color $M_1,M_2,\ldots,M_{\Delta}$ in $G$ with colors $1,2,\ldots,\Delta$, respectively. 
	\item Let $F := F_1 \cup \ldots \cup F_{\Delta}$ and color $F$ using $O(\Delta(F))$ new colors (say, via the greedy algorithm). 
	\end{enumerate}
\end{Algorithm}

The first step in analyzing $\approxcolor$ is to make sure that each call to $\rwm$ is valid, i.e., its input satisfies the premise of~\Cref{thm:random-walk-matching} (modulo the access model which we handle in the 
implementation subsection). 
This is done in the following simple claim. 

\begin{claim}\label{clm:degpeel-induction}
	In each call to $\rwm$ in Line~\eqref{line:rwm} of $\approxcolor$: 
	\begin{itemize}
	\item $\Delta(\Erem) \leq \Delta_i$; and, 
	\item $U_i$ is the set of all vertices in the \underline{entire} graph with degree $\Delta_i$ in $\Erem$. 
	\end{itemize}
\end{claim}
\begin{proof}
	The proof is by induction on $i$. For the base case of $i=1$, we have 
	\[
		\Erem = E, \quad \Delta_1 = \Delta-1+1=\Delta, \quad U_1 = V_{\Delta}.
	\]
	As $\Delta$ is the max-degree of $G$ and $V_{\Delta}$ is the set of vertices with degree $\Delta$ in $E$, the statement holds. 
	Suppose the statement holds for some iteration $i$ and we prove it for iteration $i+1$. 
	
	In iteration $i$ wherein $\Delta_i=\Delta-i+1$, the premise of~\Cref{thm:random-walk-matching} holds by induction which means that: $(1)$ $U_i$ is the set of vertices with degree $\Delta_i$ in $G$, and $(2)$ we can indeed apply~\Cref{thm:random-walk-matching} and obtain the desired matching $M_i$ incident on $U_i$.  
	Moreover, by definition, $F_i$ contains one edge per each unmatched vertex of $M_i$ in $U_i$, and thus $M_i \cup F_i$ is an edge-cover of $U_i$. 
	This means that updating $\Erem$ to $\Erem \setminus (M_i \cup F_i)$ in Line~\eqref{line:erem} ensures that in $\Erem$, the maximum degree is at most $\Delta_{i}-1 = \Delta_{i+1}$. 
	This proves the first part of the induction hypothesis for $i+1$. 
	
	For the second part, in iteration $i+1$ for $\Delta_{i+1}$, $U_i$ is chosen as the set of vertices with degree $\Delta_{i+1}$ from the set $V_{\Delta},V_{\Delta-1},\ldots,V_{\Delta_{i+1}}$  (instead of the entire $V$). 
	However, vertices in $V_1,\ldots,V_{\Delta_{i+1}-1}$ ignored here had 
	degree $<\Delta_{i+1}$  in $E$ itself to begin with and thus certainly in $\Erem \subseteq E$. Hence, in iteration $i+1$, $U_i$ is also set of all vertices in $V$ with degree $\Delta_{i+1}$, proving the second part of 
	the induction hypothesis as well.
\end{proof}

By \Cref{clm:degpeel-induction}, we can apply~\Cref{thm:random-walk-matching} to each call to the subroutine $\rwm$ in the algorithm $\approxcolor$, and we do not explicitly repeat this each time. 

The coloring output by $\approxcolor$ consists of two parts: (1) coloring the matchings $M_1,\ldots,M_{\Delta}$ with $\Delta$ separate colors, and (2) coloring $F$ with $O(\Delta(F))$ new colors; notice that by construction, 
these two sets include all edges in $E$. 
Given that $M_1,\ldots,M_{\Delta}$ are matchings, the first part clearly finds a proper $\Delta$ edge coloring of these edges. The remaining edges in $G$ are now all in $F$
and are being colored with new colors with $O(\Delta(F))$ colors and thus the coloring output by $\approxcolor$ is always a proper coloring. 

The main part of the analysis is thus to show that $\Delta(F) = O(\log{n})$ with high probability, which means the algorithm in total is going to use $\Delta+O(\log{n})$ colors as well. 
The next subsection is dedicated to this part. 

\subsection{Bounding $\Delta(F)$ in~\Cref{alg:fast-decomposition}}\label{sec:bounding-DeltaF}

For any vertex $v \in V$ and iteration $i \in [\Delta]$, we define the following random variables: 
\begin{itemize}
	\item $X_i(v) \in \set{0,1}$, where $X_i(v) = 1$ iff $v$ belongs to the set $U_i$ in iteration $i$ \underline{and} $v$ is {not} matched by $M_i$. We define $X(v) := \sum_{i = 1}^{\Delta} X_i(v)$. 
	\item $Y_i(v) \in \IN$, where $Y_i(v)$ counts the number of vertices in $U_i$ that picked the vertex $v$ as their incident neighbor in Line~\eqref{line:rnei} of~\Cref{alg:fast-decomposition}, in iteration $i$. 
	We define $Y(v) := \sum_{i=1}^{\Delta} Y_i(v)$. 
\end{itemize}
With these random variables, we now have that for every vertex $v \in V$: 
\begin{align}
	\deg_F(v) = X(v) + Y(v). \label{eq:deg-F}
\end{align}
Our goal is to bound each of these variables separately in the following two lemmas. We start with the easier case of $X$-variables, which are straightforward to bound. 

\begin{lemma}\label{lem:X-v}
	For every $v \in V$, 
	\[
		\Pr\paren{X(v) \geq 30\ln{n}} \leq \frac{1}{n^3}.
	\]
\end{lemma}
\begin{proof}
	We first have, 
	\[
		\expect{X(v)} = \sum_{i=1}^{\Delta} \expect{X_i(v)} \leq \sum_{i=1}^{\Delta}\frac{2}{\Delta_i}  = \sum_{i=1}^{\Delta} \frac{3}{\Delta-i+1} \leq 3\ln{\Delta},
	\]
	where the first inequality is by the guarantee of $\rwm$ for a vertex in $U_i$ by~\Cref{thm:random-walk-matching} (by upper bounding $1/\Delta_i + 2/\Delta_i^2$ with $3/\Delta_i$ simply), and the last one is by the standard upper bound of 
	Harmonic series. 
	
	The variables $X_1(v),\ldots,X_\Delta(v)$ are not necessarily independent of each other. However, for every $i \in [\Delta]$, regardless of the choices of $X_1(v),\ldots,X_{i-1}(v)$, 
	the variable $X_i(v)$ can only be one with probability at most $3/\Delta_i$ by~\Cref{thm:random-walk-matching} (and the independence in the randomness of different invocation of $\rwm$ across iterations). 
	This means that we can still use the upper bound obtained via Chernoff bound (\Cref{prop:chernoff}) for $X(v)$ also (as it is stochastically dominated by independent Bernoulli random variables with mean $3/\Delta_i$ for 
	the $i$-th trial). As such, 
	\[
		\Pr\paren{X(v) \geq 20\ln{n}} \leq \exp\paren{-\frac{10\ln{n}}{3}} \leq \frac{1}{n^3},
	\]
	by setting $\delta=1$ and $\mumax = 10\ln{n}$ (correct by the bound on $\expect{X(v)}$) in~\Cref{prop:chernoff}.   
\end{proof}

We now switch to bounding $Y$-variables. This is generally harder than the previous case because even a single $Y_i(v)$ depends on the choice of multiple of neighbors of $v$ which are \emph{not} independent of each other (and can even be positively correlated), 
given that the guarantee of~\Cref{thm:random-walk-matching} holds \emph{marginally} for each individual vertex (and not independently across multiple of them). 

\begin{lemma}\label{lem:Y-v}
	For every $v \in V$, 
	\[
		\Pr\paren{Y(v) \geq 270\ln{n}} \leq \frac{1}{n^{30}}.
	\]
\end{lemma}
\begin{proof}
	Fix an iteration $i \in [\Delta]$ and let $U_i(v)$ denote the neighbors $v$ in $\Erem$ in this iteration that are also in $U_i$, i.e., the \emph{max}-degree neighbors of $v$. 
	Notice that $\card{U_i(v)} \leq \Delta_i$ as the total neighbors of $v$ in this iteration is at most $\Delta_i$ (by~\Cref{clm:degpeel-induction}). 

	For every $u \in U_i(v)$, define a Bernoulli random variable $Z_i(u)$ with mean $1/\Delta_i$ (we are doing this to ``simulate'' the event of sampling a random neighbor
	for $u$ in Line~\eqref{line:rnei}, \emph{in case} $u$ becomes unmatched, and that neighbor ends up being the vertex $v$). This way, we have, 
	\[
		Y_i(v) = \sum_{u \in U_i(v)} X_i(u) \cdot Z_i(u);
	\]
	namely, to contribute to $Y_i(v)$, each vertex $u \in U_i(v)$ should remain unmatched by $M_i$, and then, after picking one of its random neighbors, it should pick $v$. Notice that with these definitions, 
	the choices of $Z_i(u)$'s are independent of each other and independent of $X_i(u)$'s; the independence is because we always pick $Z_i(u)$'s regardless of whether 
	or not $u$ actually is unmatched or not (this does not affect the value of $Y_i(v)$  because if $u$ is matched, $X_i(u) = 0$). However, we do not have any independence between $X_i(u)$ for $u \in U_i(v)$. 
	Moreover, for every $u \in U_i(v)$, 
	\begin{align}
		\expect{X_i(u)} \leq \frac{3}{\Delta_i} \quad \text{and} \quad \expect{Z_i(u)} = \frac{1}{\Delta_i}, \label{eq:X-Z-values}
	\end{align}
	where the first inequality is by~\Cref{thm:random-walk-matching} and the second equality is by the choice of $Z_i(u)$. 
	
	Putting these together, we obtain that 
	\[
		\expect{Y_i(v)} = \sum_{u \in U_i(v)} \expect{X_i(u)} \cdot \expect{Z_i(u)} \leq \Delta_i \cdot \frac{3}{\Delta_i} \cdot \frac{1}{\Delta_i} = \frac{3}{\Delta_i},  
	\]
	which in turn implies that 
	\[
		\expect{Y(v)} = \sum_{i=1}^{\Delta} \expect{Y_i(v)} \leq \sum_{i=1}^{\Delta} \frac{3}{\Delta_i} = \sum_{i=1}^{\Delta} \frac{3}{\Delta-i+1} \leq 3\ln{\Delta},
	\]
	by the upper bound on the Harmonic series. 
	
	The key part is now to prove that $Y(v)$ is also concentrated. We do this using our concentration bound in~\Cref{prop:unbounded-chernoff}. 
	Firstly, for any $i \in [\Delta]$, define $\mu_i := 3/\Delta_i$ and thus $\expect{Y_i(v)} \leq \mu_i$. Secondly and more importantly, 
	we need to bound the tail probability of each individual $Y_i(v)$. Fix any integer $y \geq 1$. We have, 
	\begin{align*}
		\Pr\paren{Y_i(v) \geq y} &= \Pr\Paren{\text{for at least $y$ choices $u \in U_i(v)$ we have $(X_i(u),Z_i(u)) = (1,1)$}} \\
		&\leq {{\Delta_i}\choose{y}} \cdot \Pr\paren{\text{for given $y$ variables: $X_i(u)=1$}} \cdot \Pr\paren{\text{for given $y$ variables: $Z_i(u)=1$}} \tag{by union bound and independence of $X$'s and $Z$'s and since $\card{U_i(v)} \leq \Delta_i$} \\
		&\leq {{\Delta_i}\choose{y}} \cdot \frac{2}{\Delta_i} \cdot \frac{1}{\Delta_i^y} \tag{by~\Cref{eq:X-Z-values} and since $Z$'s are independent of each other, but not $X$'s} \\
		&\leq \frac{e^{y} \cdot \Delta_i^{y}}{y^y} \cdot \frac{3}{\Delta_i} \cdot \frac{1}{\Delta_i^y} \tag{by the inequality ${{a}\choose{b}} \leq (a \cdot e)^b/b^b$} \\
		&= \frac{e^y}{y^y} \cdot \mu_i \tag{by the choice of $\mu_i = 3/\Delta_i$} \\
		&\leq 15 \cdot e^{-y} \cdot \mu_i;
	\end{align*}
	where the last inequality is because $e^{2y} \leq 15 \cdot y^y$, proven as follows: (1) for $y \leq e^3$, we have $e^{2y}/y^y \leq 15$ (as the maximizer of $e^{2y}/y^y$ over integers happens at $y=3$), and (2) for $y \geq e^3$, we have $e^{2y}/y^y \leq e^{2y}/e^{3y} = e^{-y}$. Notice that since $Y$ only takes integral values, we only need to prove the tail-bound for integral choices of $y \geq 1$. 
	
	Finally, similar to the proof of~\Cref{lem:X-v}, even though $Y_i(v)$ variables are not independent, they are stochastically dominated by independent random variables 
	of the same marginals given the independence of randomness of $\rwm$ across different iterations. As such, we can apply~\Cref{prop:unbounded-chernoff} to $Y(v)$ 
	with parameters $\eta = 15$, $\kappa=1$, $\mu_i = 3/\Delta_i$ for $i \in [\Delta]$, and $\mumax =  3\ln{n} \geq 3\ln{\Delta} \geq \sum_{i=1}^{\Delta} \mu_i$ (by the same calculation as that of $\expect{Y(v)}$ above) 
	to obtain that 
	\[
		\Pr\paren{Y(v) \geq 270 \cdot \ln{n}} = \Pr\paren{Y(v) \geq 6 \cdot \underbrace{\frac{15}{1^2}}_{\eta/\kappa^2} \cdot \underbrace{3\ln{n}}_{\mumax}} \leq \exp\paren{-\underbrace{15}_{\eta/\kappa} \cdot \underbrace{3 \cdot \ln{n}}_{\mumax}} = n^{-30} 
	\]
	This concludes the proof. 
\end{proof}

We can now use~\Cref{lem:X-v,lem:Y-v} to prove that $\approxcolor$ uses $\Delta + O(\log{n})$ colors with high probability. 

\begin{lemma}\label{lem:color-small-enough}
	With probability at least $1-1/n$, the number of colors used by $\approxcolor$ is at most $\Delta+300\ln{n}$. 
\end{lemma}
\begin{proof}
	Firstly, $\approxcolor$ uses $\Delta$ colors to color the matchings $M_1,\ldots,M_{\Delta}$ (which given they are edge-disjoint matchings, leads to a proper (partial) coloring). 
	Secondly, by a union bound over the events of~\Cref{lem:X-v} and~\Cref{lem:Y-v} for all vertices and by~\Cref{eq:deg-F}, 
\begin{align}
	\Pr\paren{\Delta(F) \geq 300\ln{n}} \leq \frac{1}{n}. \label{eq:final-DeltaF}
\end{align}
	Thus, in the last step, with probability $1-1/n$, the algorithm uses at most $O(\log{n})$ new colors to color the edges in $F$. Thus, $\approxcolor$ always outputs a proper edge coloring
	and only with probability $1/n$, it uses more than $\Delta+O(\log{n})$ colors. 
\end{proof}

\subsection{Implementation Details}\label{sec:fast-approx-implement}

To conclude the proof of~\Cref{thm:fast-approx}, we need to bound the runtime of $\approxcolor$ by specifying its exact implementation details. We use the following straightforward approach: 
\begin{itemize}
	\item For every vertex $v \in V$, maintain the neighbors $N(v)$ of $v$ in a hash table $H(v)$ that allows insertion, deletion, and search for $u \in N_{\Erem}(v)$ in $O(1)$ expected time. We also
	keep track of the size of $H(v)$ throughout to be able to return the degree of $v$ in $\Erem$ in $O(1)$ time. 
	
	During the algorithm, when we update $\Erem$ by deleting an edge $(u,v)$ from it, 
	we simply update $H(u)$ and $H(v)$ in $O(1)$ time as well. 
	
	The preprocessing for this step involves inserting $N(v)$ into $H(v)$ for every $v \in V$ which takes $O(m+n)$ time in expectation. 
	
	\item Partitioning of vertices into $V_1,\ldots,V_\Delta$ can be done in $O(n)$ time. 
	
	\item In each iteration $i \in [\Delta]$ of the algorithm, we can compute the set $U_i$ in $O(\card{V_{\Delta}} + \ldots + \card{V_{\Delta_i}})$ time by iterating
	over the vertices in this set and verifying their maintained degrees in $\Erem$ (through the size of their hash tables). 
	
	\item To be able to provide the access model to $\rwm$, we need to be able to sample $u$ from $N_{\Erem}(v)$ uniformly at random for $v \in U_i$, and check if $(u,v) \in \Erem$. 
	Both of these can be done trivially given we have maintained neighbors of all vertices in $\Erem$ explicitly in a hash table. 
	
	As such, by~\Cref{thm:random-walk-matching}, this step thus takes $O(\card{U_i} \cdot \log{\Delta_i})$ expected time. We further store the output matching $M_i$ in a linked-list. 
		 
	\item Finding unmatched vertices in $U_i$ and picking a random neighbor for them to add to $F_i$ takes $O(\card{U_i})$ time. Similarly, removing $M_i$ and $F_i$ from $\Erem$ 
	can be done in $O(\card{U_i})$ expected time. 
	
	All in all, we can implement each iteration $i \in [\Delta]$ of the algorithm in at most 
	\[
		O\paren{\paren{\card{V_{\Delta}}+\ldots+\card{V_{\Delta_i}}} \cdot \log{\Delta}}
	\]
	expected time (notice that $\Delta_i \leq \Delta$). Thus, the total runtime of the algorithm during the iterations is in expectation
	\[
		O(\log{\Delta}) \cdot \sum_{i=1}^{\Delta} \sum_{\substack{j=\Delta_i \\ (=\Delta-i+1)}}^{\Delta} \card{V_{j}} = O(\log{\Delta}) \cdot \sum_{k=1}^{\Delta}  k \cdot \card{V_k} = O(m\log{\Delta}), 
	\]
where the final equation is because $V_i$ is precisely the set of vertices with degree $i$ in $G$ and thus the sum is equal to the sum of degrees of vertices and hence $2m$. 
In conclusion, the total time spent by $\approxcolor$ in its $\Delta$ iterations is $O(m\log{\Delta})$ in expectation. 
	
	\item Finally, coloring $M_1,\ldots,M_{\Delta}$ with $\Delta$ colors can be done in $O(m)$ time. 
	We can run a greedy $(2\Delta(F)-1)$-coloring algorithm in the last step which, in expectation, is only 
	\[
	O(\min\set{n\log{n},m} \cdot \log\log{n})
	\]
	 by~\Cref{lem:color-small-enough}. 
	 
	 We can improve this runtime to $O(m)$ expected time by instead computing a $3\Delta(F)$ coloring as follows: 
	 for every vertex $v \in V$, maintain the set $C(v)$ of the colors used over its edges in a hash table (originally $C(v) = \emptyset$); 
iterate over the edges in an arbitrary order and when coloring an edge $e=(u,v)$, pick a random color  $c \in [3\Delta(F)]$ and check if $c$ is in $C(u)$ or $C(v)$ in $O(1)$ expected time; if not, color this 
edge and update $C(u)$ and $C(v)$ by removing $c$ from them and move to the next edge; otherwise, sample another color $c$. Given that $\card{C(u)} + \card{C(v)} \leq 2\Delta(F)$, 
we only need at most $3$ trials in expectation to color an edge. 
	
	In conclusion, the entire algorithm takes $O(m\log{\Delta})$ time in expectation. 
\end{itemize}

\subsection{Final Modifications and Concluding the Proof of~\Cref{thm:fast-approx}}\label{sec:conc-near-vizing}

We have proved so far that: 
\begin{itemize}
	\item The probability that $\approxcolor$ uses more than $\Delta+O(\log{n})$ colors is at most $1/n$ (by~\Cref{lem:color-small-enough}). 
	\item And, the expected runtime of $\approxcolor$ is $O(m\log{\Delta})$ (by running a randomized greedy $3\Delta(F)$-coloring algorithm in the last step). 
\end{itemize}
To conclude the proof of~\Cref{thm:fast-approx}, we need to ensure that the number of colors used by $\approxcolor$ is \emph{always} $\Delta+O(\log{n})$. This is straightforward to achieve via a simple modification to $\approxcolor$ (we opted to postpone this modification to the end as we find it quite minor and distracting from the main algorithm). 

We simply run $\approxcolor$ once and if it used more than the desired number of colors, we run a classical $O(m\cdot n)$ time algorithm for finding a 
$\Delta+1$ edge coloring of $G$. This way, we will always deterministically have $\Delta+O(\log{n})$ colors. Also, 
given that the probability of $\approxcolor$ failing is only at most $1/n$, the expected runtime of the second algorithm is $O(m)$ which is negligible. 

 This concludes the proof of~\Cref{thm:fast-approx}. 

\begin{remark}\label{rem:local}
	An elegant result of~\cite{Christiansen23} proved a ``Local Vizing Theorem'': every graph $G$ admits a proper edge coloring such
	that every edge $(u,v)$ receives a color from the list 
	\[
	\set{1,2,\ldots,\max\set{\deg(u),\deg(v)}+1},
	\]
	confirming a conjecture of~\cite{BonamyDLP20}. Moreover, this coloring can be found in $O(n^2\Delta)$ time. The proof of this result 
	relies heavily on the so-called \emph{multi-step Vizing chains}~\cite{Bernshteyn22}
	
	Our~\Cref{thm:fast-approx} and $\approxcolor$ implies a combinatorially weaker version of this result by coloring each edge $(u,v)$ from the list 
	\[
		\set{1,2,\ldots,\max\set{\deg(u),\deg(v)}+O(\log{n})},
	\]
	but algorithmically faster and only in $O(m\log{\Delta})$ expected time\footnote{This is because in~$\approxcolor$, each $M_i$ for $i \in [\Delta]$ can be colored with the color $i + O(\log{n})$ (which belongs to the list of its edges given one endpoint has 
	degree $i$) and $F$ can be colored with colors in $O(\log{n})$.}. Additionally, as was apparent, the proof relies on no Vizing chains and is inherently different from prior approaches in this context. 
\end{remark}

\clearpage


\section{A Near-Quadratic Time Algorithm for Vizing's Theorem}\label{sec:vizing}

As an almost immediate corollary of~\Cref{thm:fast-approx} and standard techniques, we can also obtain an algorithm for Vizing's theorem that runs in near-quadratic time in expectation. 

\begin{theorem}\label{thm:fast-vizing}
	There is a randomized algorithm that given any simple graph $G=(V,E)$ with maximum degree $\Delta$, outputs a 
	proper $(\Delta+1)$ edge coloring of $G$ in $O(n^2\log{n})$ time in expectation.
\end{theorem}

Recall that by a trivial success amplification---namely, running the algorithm $O(\log{n})$ times and truncating it if takes more than twice its expected runtime---we can improve the 
guarantee in this theorem to a high probability bound at a cost of increasing the runtime by an $O(\log{n})$ factor. 

To prove~\Cref{thm:fast-vizing}, we need a standard result for edge coloring partially colored graphs: 
an algorithm that can extend a partial proper $(\Delta+1)$ edge coloring of a graph by coloring one more edge (and possibly recoloring some already-colored edges differently). 
This result dates back to the original proof of Vizing's theorem~\cite{Vizing64} using the so-called \emph{Vizing Fans} and \emph{Vizing Chains}. 

\begin{proposition}[cf.~\cite{Vizing64,RaoD92,MisraG92}]\label{prop:det-extend}
	There is a deterministic algorithm with the following properties. Let $G=(V,E)$ be any undirected simple graph and $C$ be a proper $(\Delta+1)$ edge coloring of a subset $F \subseteq E$ of edges. Let $e \in E \setminus F$ be any arbitrarily uncolored edge. 
	Then, given $G$ (via adjacency list access), $F$, $C$, and $e$, the algorithm in $O(n)$ time extends the coloring of $C$ to a proper $(\Delta+1)$ edge coloring of $F \cup \set{e}$ (by possibly recoloring some edges in $F$ as well). 
\end{proposition}

Equipped with~\Cref{prop:det-extend} and our own~\Cref{thm:fast-approx}, we can prove~\Cref{thm:fast-vizing} quite easily. 

\begin{proof}[Proof of~\Cref{thm:fast-vizing}]
	The algorithm runs in two separate steps. 

\smallskip

	\textbf{First step:} First run our algorithm in~\Cref{thm:fast-approx} to obtain a $\Delta+O(\log{n})$ coloring of $G$ in expected $O(m\log{n})$ time. Then, pick the $\Delta+1$ largest color classes (in terms of the number of edges they color) 
	among these $\Delta+O(\log{n})$ colors and let $C$ be a proper $(\Delta+1)$ edge coloring of these subset of edges. Remove all remaining colors. 
	
	At this point, the number of uncolored edges by $C$ is: 
	\[
		\frac{O(\log{n})}{\Delta+O(\log{n})} \cdot m = O(\frac{m\log{n}}{\Delta}).  
	\]
	
	\textbf{Second step:} Run~\Cref{prop:det-extend} on each of the uncolored edges of $C$ one at a time to extend the $(\Delta+1)$ edge coloring of the first step to the entire graph in 
	\[
		O\paren{\frac{m \cdot n \cdot \log{n}}{\Delta}} = O\paren{n^2\log{n}},
	\]
	time deterministically. 
	
	In conclusion, after these two steps, we obtain a proper $(\Delta+1)$ edge coloring of the entire graph in $O(n^2\log{n})$ expected time, 
	concluding the proof.  
\end{proof}
\begin{remark}
As is apparent from the proof, the runtime of the algorithm in~\Cref{thm:fast-vizing} is actually 
\[
	O(\frac{m \cdot n}{\Delta} \cdot \log{n})
\]
in expectation. This can be  considerably faster than the advertised $\Ot(n^2)$ time, whenever the graph is ``far from`` being $\Delta$-regular, i.e., when $m$ is considerably smaller than $n \cdot \Delta$. 
For instance, on graphs with arboricity $\alpha$, where $m = O(n \cdot \alpha)$ instead, this translates to a runtime of $O(n^2\log{n} \cdot \alpha/\Delta)$. This should be compared
with the $\Ot(m\sqrt{n} \cdot \alpha/\Delta)$ runtime of~\cite{BhattacharyaCPS23}. 
\end{remark}

\section{A Linear-Time Algorithm for $(1+\epsilon) \Delta$ Edge Coloring } \label{sec:approximate}

We can also show that a small modification of our algorithms can be used to 
obtain an $O(m \cdot \log{(1/\eps)})$ time randomized algorithm for the $(1+\eps)\Delta$ edge coloring problem. 

\begin{theorem}\label{thm:fast-eps-approximate}
	The following holds for some absolute constant $\eta_0 \in \IN$. 
	There is a randomized algorithm that given any simple graph $G=(V,E)$ with maximum degree $\Delta$ and $\eps \geq \eta_0 \cdot \frac{\log{n}}{\Delta}$, outputs a 
	proper $(1+\eps)\,\Delta$ edge coloring of $G$ in $O(m \cdot \log{(1/\eps)})$ time in expectation.
\end{theorem}

The algorithm in~\Cref{thm:fast-eps-approximate} is obtained via a minor modification of the algorithm in~\Cref{thm:fast-approx}. In fact, we only need to change the budget
in the subroutine $\rwm$ for finding fair matchings to be $\ln{(1/\eps)}$ (instead of $2\ln{\Delta}$) and then use the modified $\rwm$ directly in~\Cref{alg:fast-decomposition}. 

The following lemma captures the properties of $\rwm$ after these modifications. 
\begin{lemma}\label{lem:new-rwm}
	Given an input $(V,E,V_{\Delta},V)$ with the access model of~\Cref{thm:random-walk-matching}, the modified $\rwm$ algorithm outputs a matching $M$ such that for every $v \in V_{\Delta}$, 
	\[
		\Pr\paren{\text{$v$ is unmatched $M$}} \leq \frac{1}{\Delta} + \frac{1}{\Delta^2} + \eps. 
	\]	
	Moreover, the modified $\rwm$ runs in $O(n_{\Delta} \cdot \log{(1/\eps)})$ time. 
\end{lemma}
\begin{proof}
	To prove the probability bound, we follow the same exact proof as in~\Cref{lem:prob-unmatched} with the following modifications. 
	
	Firstly, by modifying the budget of $\rwm$ to be $\ln{(1/\eps)}$, in~\Cref{eq:v-unmatched-1}, we obtain 
	\[
		\Pr\paren{\text{$v$ was never matched during $\rwm$}} \leq \exp\paren{-b} \leq \eps.  
	\]
	Also, we have that 
	\begin{align*}
		&\sum_{i=1}^{k}\Pr\paren{\text{$v$ got unmatched by $M_i \triangle P_i$ and never got matched by any $M_j$ for $j > i$}} \\
		&\leq \frac{1}{\Delta} + \frac{1}{\Delta^2},
	\end{align*}
	as this step did not depend on the budget of the algorithm. These two together imply, exactly as in~\Cref{lem:prob-unmatched}, that, 
	\[
		\Pr\paren{\text{$v$ is unmatched $M$}} \leq \eps + \frac{1}{\Delta} + \frac{1}{\Delta^2},
	\]
	as desired. 
	
	As for the runtime, we follow the same exact implementation in~\Cref{sec:implementation}. Thus, by~\Cref{eq:runtime-extract},
	we have
	\[
	\Exp_{W,R}\Bracket{\text{runtime of $\extract(W)$}} = O(1) \cdot \Exp\card{W}, 
	\]
	as again no modification was done here. Finally, by following the proof of~\Cref{lem:rwm-runtime}, the expected length of all Matching Random Walks in the modified $\rwm$ is now,
	\[
		\Exp\Bracket{\text{length of all Matching Random Walks}} \leq 7n_{\Delta} \cdot \underbrace{\ln{(1/\eps)}}_{=\text{budget $b$}}.
	\]
	Putting all these together, the modified $\rwm$ runs in $O(n_{\Delta} \cdot \log{(1/\eps)})$ expected time. 
\end{proof}
We can now prove~\Cref{thm:fast-eps-approximate}. 

\begin{proof}[Proof of~\Cref{thm:fast-eps-approximate}]
Running the modified $\rwm$ inside $\approxcolor$ will mean that in each step, by~\Cref{lem:new-rwm}, the probability that a vertex $v$ remains unmatched is at most $2\eps$ (as $\eps \gg 1/\Delta+1/\Delta^2$ by the assumption in the theorem).  
Using the same analysis as in~\Cref{lem:X-v} and~\Cref{lem:Y-v}, we have, 
\[
	\Pr\paren{\Delta(F) \geq 1000\eps\Delta} \leq n^{-10}, 
\]
since $\expect{\deg_F(v)} \leq 2\eps\Delta$ for every vertex and our assumption on $\eps$ makes this quantity $\Theta(\log{n})$ for some large hidden constant and allows us to apply the desired concentration bounds in~\Cref{prop:chernoff} and~\Cref{prop:unbounded-chernoff} used in establishing~\Cref{lem:X-v} and~\Cref{lem:Y-v}. 

Finally, finding an $O(\Delta(F))$ coloring in the last step uses $O(\eps\Delta)$ more colors. Thus, this algorithm uses $\Delta + O(\eps) \cdot \Delta$ colors. By re-parameterizing $\eps \leftarrow \Theta(\eps)$, 
we obtain a $(1+\eps)\Delta$ coloring as desired. The runtime is also $O(m\log{(1/\eps)})$ given by~\Cref{lem:new-rwm}. 
\end{proof}

\section*{Acknowledgement} 
\addcontentsline{toc}{section}{Acknowledgement}

I would like to thank Lap Chi Lau for insightful comments that helped with the presentation of the paper, 
and Eva Rotenberg for helpful pointers on ``Local Vizing Theorem'' in~\cite{Christiansen23} and its connection to~\Cref{rem:local} in this paper. 
I am also grateful to David Harris, Michael Kapralov, Sanjeev Khanna, and David Wajc for illuminating discussions.

\bibliographystyle{halpha-abbrv}
\bibliography{new}

\clearpage
\appendix

\part*{Appendix}


\section{Finalizing the Runtime Analysis in~\Cref{thm:random-walk-matching}: Proof of~\Cref{eq:runtime-extract}}\label{app:implementation-matching} 

We finalize the runtime analysis of~$\rwm$ here by proving~\Cref{eq:runtime-extract}, i.e., showing that
\[
	\Exp_{W,R}\Bracket{\text{runtime of $\extract(W)$}} = O(1) \cdot \Exp\card{W}.
\]
We follow the two-step strategy outlined in~\Cref{sec:implementation} for this proof. 
\begin{lemma}\label{lem:extract-time}
	Over the randomness of $W$ and inner randomness $R$ of the algorithm $\extract$, 
	\[
		\Exp_{W,R}\Bracket{\text{runtime of $\extract(W)$}} = O(1) \cdot \paren{\expect{a(W)} + \frac{1}{\sqrt{\Delta}} \cdot \Exp\bracket{a(W)\cdot\log^4{(a(W))}}}, 
	\]
	where $a(W)$ is the number of the iterations run in $\extract(W)$ before it terminates. 
\end{lemma}
\begin{proof}
	As argued earlier, $\extract$ processes the Matching Random Walk in a one-way fashion meaning that when processing the vertex $v_{2i-1}$, the choice of $v_{2i} \in N(v_{2i-1})$ is still uniformly random. We shall use this property
	crucially in this proof also. 
	
	We prove that in each step $i \in \floor{(\card{W}-1)/2}$ of $\extract(W)$, the runtime of the algorithm for processing $v_{2i-1}$ is $O(1)$ \emph{amortized expected} time over the randomness of $v_{2i}$ from $N(v_{2i-1})$ as well 
	as the inner randomness of $\extract$. 
	
	To avoid clutter, in the following, we denote $u := v_{2i-1}$ and $v := v_{2i}$ chosen uniformly at random from $N(u)$. 
	Define: 
	\[
		even(u) := \text{\# of `even'-labeled in $N(u) \cap P$}, \quad odd(u) := \text{\# of `odd'-labeled in $N(u) \cap P$}. 
	\]
	To perform the amortized analysis, we assume for each arriving $u$ processed in $\extract$, we gain $\Theta(1)$ tokens 
	which is assigned to this vertex. 
	
	The expected amortized runtime of processing $u$ is then: 
	\begin{align}
		&\hspace{15pt} \Theta(1) \label{eq:C1} \\
		&+ \frac{even(u)+odd(u)}{\Delta} \cdot \expect{\text{amortized time of checking the label of $u$}} \label{eq:C2}\\
		&\hspace{0pt} + \frac{even(u)}{\Delta} \cdot \expect{\text{amortized time of $\fixeven(u)$}} \label{eq:C3} \\
		&+ \frac{odd(u)}{\Delta} \cdot \expect{\text{amortized time of $\fixodd(u)$}} \label{eq:C4}. 
	\end{align}
	
	\paragraph{Runtime of~\eqref{eq:C2}.} This step involves spending $O(\log{\card{T}} + \card{S})$ time. On the other hand, we can use 
	one of the tokens per each vertex in $S$ to ``gain'' additional $\Theta(\card{S})$ time (and the vertex is moved from $S$ so this can only happen once). Thus, the amortized expected runtime here is
	\[
		O(1) \cdot \frac{even(u)+odd(u)}{\Delta} \cdot \paren{\log{\card{T}} + \card{S} - \card{S}} = O(1) \cdot \frac{even(u)+odd(u)}{\Delta} \cdot \log{\card{T}}. 
	\]
	\paragraph{Runtime of~\eqref{eq:C3}.}  This step involves running $\emph{Delete(T,$2j+2,2k+1$)}$ which takes $O(\log{\card{T}})$ time. On the other hand, we are also deleting $k-j$ vertices from $P$ (and $T$).
	Here, we use the tokens per each vertex deleted from $P$ to ``gain'' additional $\Theta(k-j)$ time (and the vertex is removed entirely and so this can only happen once). Thus, the amortized expected runtime here is
	\[
		O(1) \cdot \frac{even(u)}{\Delta} \cdot \paren{\log{\card{T}}- (k-j)};
	\]
	note that in this step, we are actually \emph{gaining} $\Theta(k-j)$ time in an amortized sense. 
	
	We now calculate how large $k-j$ can be in expectation. Conditioned on $v$ being one of the `even'-labeled neighbors of $u$ in $P$, the choice of $v$ is  
	uniform at random among $even(u)$ many choices. Thus, 
	\[
		\expect{k-j \mid \text{$v$ is `even'-labeled}} \geq \frac{even(u)}{2};
	\]
	namely, even if the choices of $v$ are all the $even(u)$ vertices closest to $u$ in the path $P$, in expectation the distance between $k$ and $j$ is at least $even(u)/2$. 
	So, the expected amortized runtime of here is:
	\[
		O(1) \cdot \frac{even(u)}{\Delta} \cdot \paren{\log{\card{T}} - \frac{even(u)}{2}}. \label{eq:expect-even}
	\]
	\paragraph{Runtime of~\eqref{eq:C4}.} This step involves (a) running \emph{Pred} for $O(j-\ell)$ time (in an iterator) which takes $O(\log{\card{T}} + j-\ell)$ time; (b) sample $\log\card{T}$ vertices, checking them against 
	the hash table $H$ for $O(1)$ expected time, and running \emph{Search} for each one that is in $P$ in $O(\log\card{T})$ expected time each; (c) running \emph{Delete} and \emph{Reverse} on two sequences
	which take yet another $O(\log\card{T})$ time. At the same time, we now have also removed $O(j-\min\set{\ell,\ell'})$ vertices from $P$ and can use the tokens such vertex. Thus, 
	the amortized expected runtime here is 
	\[
		O(1) \cdot \frac{odd(u)}{\Delta} \cdot \paren{\log{\card{T}} + (j-\ell) + \sum_{u': \text{sampled in $P$}} \log{\card{T}} - 2 \cdot (j-\min\set{\ell,\ell'})};
	\]
	again, in this step, we are gaining $\Theta(j-\min\set{\ell,\ell'})$ time in an amortized sense. 
	
	We now calculate the value of $j-\min\set{\ell,\ell'}$. The probability that none of the $\log{\card{T}}$ samples lands in an `odd'-labeled neighbor of $u$ is 
	\[
		\paren{1-\frac{odd(u)}{\Delta}}^{\log\card{T}} \leq \exp\paren{-\frac{odd(u)}{\Delta} \cdot \log\card{T}} \leq 1-\min\set{\frac{odd(u) \cdot \log{\card{T}}}{2\Delta},\frac{1}{e}},
	\]
	where used the inequalities $1-x \leq e^{-x} \leq 1-x/2$ for $x \in [0,1]$. Conditioned on having a sample $u'$ from `odd'-labeled neighbors of $u$, the distribution of $v$ and $u'$ (which are independent of each other) is 
	that of two independent `odd'-labeled neighbors of $u$. Even if all these vertices are next to each other in $P$, we have $odd(u)$ many choices and thus in this case, 
	\[
		\expect{j-\ell' \mid \text{$v$ is `odd'-labeled and an `odd'-labeled neighbor $u'$ of $u$ in $P$ is sampled}} \geq \frac{odd(u)}{3}. 
	\]
	Thus, the amortized expected runtime of this step is
	\[
	O(1) \cdot \frac{odd(u)}{\Delta} \paren{\log{\card{T}} + \underbrace{\frac{odd(u)+even(u)}{\Delta}}_{=\Pr\paren{u' \in P}} \cdot \log^2{\card{T}} - \min\set{\frac{odd(u) \cdot \log{\card{T}}}{2\Delta},\frac{1}{e}} \cdot \frac{odd(u)}{3}},
	\]
	where we used the independence of each of the $\log\card{T}$ samples and the choice of $v$ from $N(u)$. 
	
	\paragraph{Combining everything together.} We now compute the final runtime by considering three different cases. 
	\begin{itemize}
		\item \emph{Case 1 -- both $odd(u),even(u) \leq \Delta/\log{\card{T}}$:} in this case, the contributions of~\eqref{eq:C2}~\eqref{eq:C3}, and~\eqref{eq:C4} are asymptotically at most (by dropping all negative terms when we are ``gaining'' time):
		\begin{align*}
		 \paren{\frac{even(u)+odd(u)}{\Delta} \cdot \log{\card{T}}} + \frac{even(u)}{\Delta} \cdot \paren{\log{\card{T}}} + \frac{odd(u)}{\Delta} \paren{\log{\card{T}} + \frac{odd(u)+even(u)}{\Delta} \cdot \log^2{\card{T}}},
		\end{align*}
		which is $O(1)$ given $odd(u)+even(u) \leq \Delta/\log{\card{T}}$. 
		\item \emph{Case 2 -- $even(u) \geq \Delta/\log{\card{T}}$:} in this case, the contributions of~\eqref{eq:C2}~\eqref{eq:C3}, and~\eqref{eq:C4} are asymptotically at most (by dropping negative terms in~\eqref{eq:C4} but not~\eqref{eq:C3}): 
		\begin{align*}
		 &\paren{\frac{even(u)+odd(u)}{\Delta} \cdot \log{\card{T}}} + \frac{even(u)}{\Delta} \cdot \paren{\log{\card{T}}}  - \frac{even(u)}{\Delta} \cdot \paren{\frac{even(u)}{2}}  \\
		 &\hspace{20pt} + \frac{odd(u)}{\Delta} \paren{\log{\card{T}} + \frac{odd(u)+even(u)}{\Delta} \cdot \log^2{\card{T}}}\\
		 &\leq 3\log{\card{T}} + \log^2{\card{T}} - \frac{even(u)}{\Delta} \cdot \paren{\frac{even(u)}{2}} \tag{by using $odd(u)+even(u) \leq \Delta$} \\
		 &\leq 4\log^2{\card{T}}  - \frac{1}{\log{\card{T}}} \cdot \paren{\frac{\Delta}{2\log{\card{T}}}} \tag{by the lower bound on $even(u)$} \\
		 &\leq \frac{16\log^4{\card{T}}}{\sqrt{\Delta}} \tag{by considering $8\log^2{\card{T}} \geq \Delta$ case as otherwise the RHS is negative}. 
		\end{align*}
		\item \emph{Case 3 -- $odd(u) \geq \Delta/\log{\card{T}}$}:  in this case, the contributions of~\eqref{eq:C2}~\eqref{eq:C3}, and~\eqref{eq:C4} are asymptotically at most (by dropping negative terms in~\eqref{eq:C3} but not~\eqref{eq:C4}):
		\begin{align*}
			&\paren{\frac{even(u)+odd(u)}{\Delta} \cdot \log{\card{T}}} + \paren{\frac{even(u)}{\Delta} \cdot {\log{\card{T}}}} \\
			&\hspace{20pt} + \frac{odd(u)}{\Delta} \paren{\log{\card{T}} + \frac{odd(u)+even(u)}{\Delta} \cdot \log^2{\card{T}} - \min\set{\frac{odd(u) \cdot \log{\card{T}}}{2\Delta},\frac{1}{e}} \cdot \frac{odd(u)}{3}}, \\
			&\leq 3\log{\card{T}} + \log^2{\card{T}} - \frac{odd(u)}{\Delta} \cdot \frac{odd(u)}{12} \tag{by using $odd(u) + even(u) \leq \Delta$ and $odd(u) \cdot \log{\card{T}} \geq \Delta$ for the min-term} \\
			&\leq 4\log^2{\card{T}} - \frac{\Delta}{12\log^2{\card{T}}} \tag{by the lower bound on $odd(u)$} \\
			&\leq \frac{96 \log^4{\card{T}}}{\sqrt{\Delta}} \tag{by considering $48\log^2{\card{T}} \geq \Delta$ case as otherwise the RHS is negative}. 
		\end{align*}
	\end{itemize}
	
	Thus, the expected time of a step, conditioned on the choice of the walk $W$ up until here is
	\[
		O(1) + \frac{O(\log^4{\card{T}})}{\sqrt{\Delta}}.	
	\]
	Moreover, size of $T$ is at most equal to the number of iterations that $\extract(W)$ is run (up to constant factors). Putting all of these together over the $a(\card{W})$ steps of $\extract(W)$ (which itself is a random variable), we obtain
	\begin{align*}
		\Exp_{W,R}\Bracket{\text{runtime of $\extract(W)$}} &= \expect{a({W}) \cdot \paren{O(1) + \frac{O(\log^4{({a(W)})})}{\sqrt{\Delta}}}} \\
		&= O(1) \cdot \paren{\expect{a(W)} + \frac{1}{\sqrt{\Delta}} \cdot \Exp\bracket{a({W})\cdot\log^4{({a(W))}}}},
	\end{align*}
	which concludes the proof. 
\end{proof}

Consider the RHS of~\Cref{lem:extract-time}. By~\Cref{lem:rwm-runtime}, we have a handle on the total length of $\expect{a(W)} \leq \Exp\card{W}$ across all the walks in this term and can bound it by $O(n_{\Delta}\ln{\Delta})$ as desired. 
However, there is also an additional term of $\Delta^{-1/2} \cdot \expect{a(W)\cdot\log^4{({a(W)})}}$ here. But, the function $f(x) = x\log^4{x}$ is strictly convex and thus we \emph{cannot} bound
$E[f(x)]$ by any function of $f(E[x])$ (say, via Jensen's inequality) directly, without more knowledge about the distribution of $x$.\footnote{Of course, in our case, we could have bounded this with term with 
$\expect{W} \cdot O(\log^4{n})$ as the treap $T$ is storing a path and thus never becomes larger than $O(n)$. This means that 
as long as $\Delta \geq \log^{8}{n}$, we will still achieve the desired runtime. But, this will further increase the runtime of our final edge coloring algorithm by an additive factor of $O(n\log^{9}(n))$ time. 
A more careful analysis (and some changes to the implementation step) can reduce this to $\Delta \geq \log^2{n}$, but we do not take that route since the subsequent analysis requires 
no assumption on $\Delta$.}

In the following lemma, we bound this term  via a simple concentration result on the ``effective'' length of $W$ (when running $\extract(W)$) using our random truncation in Line~\eqref{line:random-terminate} of $\extract$. 

\begin{lemma}\label{lem:inverse-jensen}
	Over the randomness of $W$ and inner randomness $R$ of the algorithm $\extract$, 
	\[
		\expect{a({W})\cdot\log^4{(a({W}))}} \leq O(\log^4{\Delta}) \cdot \expect{a(W)}.
	\]
\end{lemma}
\begin{proof}
	Recall that $a(W)$ is the number of steps $\extract(W)$ is being run before it terminates. By~Line~\eqref{line:random-terminate} of $\extract$, the algorithm terminates in every step (except for the first one) with probability $1/\Delta^2$. Thus, 
	for any $t \geq 1$, 
	\begin{align}
		\Pr\paren{a(W) \geq t \cdot \Delta^2} \leq \paren{1-\frac{1}{\Delta^2}}^{t \cdot \Delta^2} \leq e^{-t}. \label{eq:Delta2-t}
	\end{align}
	Moreover, the random variable $a({W})\cdot\log^4{(a({W}))}$ is integer-valued by considering $\log^4(\cdot)$ as taking the ceiling as well (this can only increase the value of the expectation). 
	Thus, we have, 
	\begin{align*}
		\expect{a({W})\cdot\log^4{(a({W}))}} &= \sum_{x \geq 1} \Pr\paren{a({W})\cdot\log^4{(a({W}))} \geq x} \tag{as $\expect{X} = \sum_{i \geq 1} \Pr\paren{X \geq i}$ for any integer-valued (non-negative) $X$} \\
		&\leq \sum_{x=1}^{\Delta^{3}} \Pr\paren{a({W})\cdot\log^4{(a({W}))} \geq x} + \sum_{x \geq \Delta^{3}} \Pr\paren{a({W})\cdot\log^4{(a({W}))} \geq x} \tag{by splitting the sum over the first $\Delta^3$ terms} \\
		&\leq \sum_{x=1}^{\Delta^{3}} \Pr\paren{a({W}) \cdot \log^4{(\Delta^{3})} \geq x} + \sum_{x \geq \Delta^{3}} \Pr\paren{a({W}) \geq \frac{x}{\log^4{x}}},
		\intertext{where we used $\log^4{a(W)} \leq \log^{4}{(\Delta^{3})}$ for this range of $x$ in the first term and that for the second term, the event $a({W})\log^4{(a({W}))} \geq x$ is subsumed by the event $a({W}) \geq \frac{x}{\log^4{x}}$ (consider
		the case $x \geq a(W)$ and $x< a(W)$ separately).  Continuing the above equations, we have, } 
		&\leq \sum_{x=1}^{\infty}  \Pr\paren{a({W}) \cdot \log^4{(\Delta^{3})} \geq x} + \sum_{x \geq \Delta^{3}} \Pr\paren{a({W}) \geq \frac{x}{\log^4{x}}}, \tag{by letting $x$ goes to infinity in the first series} \\
		&= \expect{a({W}) \cdot \log^4{(\Delta^{3})}} + \sum_{x \geq \Delta^{3}} \Pr\paren{a({W}) \geq \frac{x}{\log^4{x}}} \tag{similar to the first step for the random variable $X := a(W) \cdot \log^4{(\Delta^{3})}$} \\
		&\leq O(\log^4{\Delta}) \cdot \expect{a(W)} + \sum_{x \geq \Delta^{3}} \exp\paren{-\frac{x}{\log^4{(x)} \cdot \Delta^2}} \tag{by~\Cref{eq:Delta2-t}} \\
		&= O(\log^4{\Delta}) \cdot \expect{a(W)} + O(1), 
	\end{align*}
	since the second term forms a geometric series with the first term less than $1$ already for $x \geq \Delta^3$ (for large enough constant $\Delta$, otherwise, there will be $O(1)$ terms at the beginning which are larger than $1$, but still a constant). 
	This concludes the proof as $\expect{a(W)} \geq 1$. 
\end{proof}

Combining~\Cref{lem:extract-time} with~\Cref{lem:inverse-jensen}, plus the fact that $a(W) \leq \card{W}$, we conclude:  
\[
	\Exp_{W,R}\Bracket{\text{runtime of $\extract(W)$}} \leq O(1) \cdot \paren{\expect{a(W)} + \frac{\log^4{\Delta}}{\sqrt{\Delta}} \cdot \Exp\bracket{a(W)}}  = O(1) \cdot \Exp\card{W},
\]
which proves~\Cref{eq:runtime-extract} as desired.


\section{A Concentration Inequality for Unbounded Variables}\label{app:concentration} 

We prove our simple concentration inequality in~\Cref{prop:unbounded-chernoff} here. 

\begin{proposition*}[Restatement of~\Cref{prop:unbounded-chernoff}]
	Let $\mu_1,\ldots,\mu_n$ be real non-negative numbers and $X_1,\ldots,X_n$ be independent non-negative random variables such that for every $i \in [n]$:
	\[
		\expect{X_i} \leq \mu_i \quad \text{and} \quad \forall x \geq 1: ~ \Pr\paren{X_i \geq x} \leq \eta \cdot e^{-\kappa \cdot x} \cdot \mu_i, 
	\]
	for some $\eta, \kappa >0$. Define $X := \sum_{i=1}^{n} X_i$ and let $\mumax \geq \sum_{i=1}^{n} \mu_i$. Then, 
	\[
		\Pr\paren{X \geq \frac{6\eta}{\kappa^2} \cdot \mumax} \leq \exp\paren{-\frac{\eta}{\kappa} \cdot \mumax}. 
	\]
\end{proposition*}
\begin{proof}
	We use a standard moment generating function approach. Let $s := \kappa/2$. For any $x > 0$, 
	\begin{align*}
		\Pr\paren{X \geq x} &= \Pr\paren{e^{s \cdot X} \geq e^{s \cdot x}} \tag{as $X \geq x$ is equivalent to $e^{s \cdot X} \geq e^{s \cdot x}$} \\
		&\leq e^{-s \cdot x} \cdot \expect{e^{s \cdot X}} \tag{by Markov bound} \\
		&= e^{-s \cdot x} \cdot \prod_{i=1}^{n} \expect{e^{s \cdot X_i}} \tag{since $X_1,\ldots,X_n$ are independent of each other}. 
	\end{align*}
	We can now bound each term for $i \in [n]$ as follows:
	\begin{align*}
		\expect{e^{s \cdot X_i}} &= \Pr\paren{X_i=0} \cdot 1 + \int_{\theta=1}^{\infty} e^{s \cdot \theta} \cdot \Pr\paren{X_i = \theta} \, d \theta \tag{by definition, as $X_i$'s are non-negative} \\
		&\leq 1 + \int_{\theta=1}^{\infty} e^{s \cdot \theta} \cdot \eta \cdot e^{-\kappa \cdot \theta} \cdot \expect{X_i} \, d \theta \tag{by the given upper bound in the statement} \\
		&= 1 + \eta \cdot \expect{X_i} \cdot \int_{\theta=1}^{\infty} e^{-(\kappa/2) \cdot \theta}  \, d \theta \tag{by re-arranging the terms and using $s=\kappa/2$} \\
		&= 1 + \left. \eta \cdot \expect{X_i} \cdot \frac{2}{\kappa} \cdot e^{-(\kappa/2) \cdot \theta}  \right]_{\theta=1}^{\infty} \tag{as the integral of $f(x) = e^{\alpha x}$ is $e^{\alpha x}/\alpha$} \\
		&= 1 + \eta \cdot \expect{X_i} \cdot \frac{2 \cdot e^{-(\kappa/2)}}{\kappa} \tag{as $\lim_{\theta \rightarrow \infty} e^{-(\kappa/2) \cdot \theta} = 0$} \\
		&\leq 1 + \eta \cdot \expect{X_i} \cdot \frac{2}{\kappa} \tag{as $e^{-(\kappa/2)} \leq 1$} \\
		&\leq \exp\paren{\expect{X_i} \cdot \frac{2\eta}{\kappa}} \tag{as $1+y \leq e^{y}$ for all $y > 0$}.  
	\end{align*}
	Plugging in this bound in the previous equation, we get that
	\begin{align*}
		\Pr\paren{X \geq x} &\leq e^{-s \cdot x} \cdot \prod_{i=1}^{n} \exp\paren{\expect{X_i} \cdot \frac{2\eta}{\kappa}} \\
		&= e^{-s \cdot x} \cdot  \exp\paren{\expect{X} \cdot \frac{2\eta}{\kappa}} \\
		&=\exp\paren{- \paren{\frac{\kappa}{2} \cdot x - \frac{2\eta}{\kappa} \cdot \expect{X}}}. \tag{by the choice of $s=\kappa/2$}
	\end{align*}
	Finally, plugging the value of $x=(6\eta/\kappa^2) \cdot \mumax$, we get 
	\begin{align*}
		\Pr\paren{X \geq \frac{6\eta}{\kappa^2} \cdot \mumax} &\leq \exp\paren{- \paren{\frac{\kappa}{2} \cdot \frac{6\eta}{\kappa^2} \cdot \mumax - \frac{2\eta}{\kappa} \cdot \expect{X}}} \\
		&\leq \exp\paren{-\frac{\eta}{\kappa} \cdot \mumax},
	\end{align*}
	concluding the proof. 
\end{proof}

\begin{remark}
The setting of~\Cref{prop:unbounded-chernoff} is quite similar-in-spirit to Bernstein's inequality for sum of sub-exponential variables but with a multiplicative guarantee instead of an additive one.
However, still, it does not provide a close-to-expectation concentration inequality (as in $X$ being $(1+o(1)) \cdot \expect{X}$ almost surely).
Rather, it only shows that as long as $\eta,\kappa=\Theta(1)$, $X$ cannot be more than a (large) constant factor away from $\expect{X}$ with high enough probability. 

For our purpose, this was already enough
and hence we did not attempt to further optimize the constants and the multiplicative guarantee of this result. 
\end{remark}

\end{document}